\documentclass[11pt]{article}
\usepackage{a4wide}
\usepackage[usenames,dvipsnames]{color}
\usepackage{amssymb,amsmath,graphics,color,enumerate,eucal}
\usepackage{tikz}
\usepackage{xspace}
\usetikzlibrary{arrows,shapes,snakes,automata,backgrounds,petri,patterns} 
\usepackage{todonotes}
\usepackage[utf8x,utf8]{inputenc} 
\usepackage[T1]{fontenc}
\usepackage[vlined,linesnumbered,boxruled]{algorithm2e}

\usepackage{standalone}
\usepackage{tikz}
\usetikzlibrary{arrows}
\usetikzlibrary{patterns}

\usepackage{epsfig}
\usepackage{subfig}
\usepackage{amsmath,amsfonts,amssymb,epsfig,color,amsthm}
\usepackage{comment}
\usepackage{thmtools}
\usepackage{thm-restate}
\usepackage{mdframed}

\newtheorem{theorem}{Theorem}[]
\newtheorem{lemma}[theorem]{Lemma}
\newtheorem{corollary}[theorem]{Corollary}

\newtheorem{proposition}[theorem]{Proposition}

\newtheorem{conjecture}{Conjecture}
\newcommand{\ceil}[1]{\ensuremath{\left\lceil{#1}\right\rceil}}%
\newcommand{\ignore}[1]{}%

\newcommand{\ProblemFormat}[1]{{\sc #1}}
\newcommand{\ProblemName}[1]{\ProblemFormat{#1}\xspace}

\newcommand{\probTSAT}[0]{\ProblemName{$3$-SAT}}
\newcommand{\probSteinerRC}{\ProblemName{Steiner Rainbow $k$-Coloring}}
\newcommand{\probPartRC}{\ProblemName{Subset Rainbow $k$-Coloring}}
\newcommand{\probPartRTwoCExt}{\ProblemName{Subset Rainbow $2$-Coloring Extension}}
\newcommand{\probPartRCExt}{\ProblemName{Subset Rainbow $k$-Coloring Extension}}
\newcommand{\probPartRTwoC}{\ProblemName{Subset Rainbow $2$-Coloring}}
\newcommand{\probMaxPartRC}{\ProblemName{Maximum Subset Rainbow $k$-Coloring}}
\newcommand{\probMaxRC}{\ProblemName{Maximum Rainbow $k$-Coloring}}
\newcommand{\probMaxRTwoC}{\ProblemName{Maximum Rainbow $2$-Coloring}}
\newcommand{\probRC}{\ProblemName{Rainbow $k$-Coloring}}
\newcommand{\probRTwoC}{\ProblemName{Rainbow $2$-Coloring}}

\newcommand{\eps}{\varepsilon}

\newcommand{\Var}{{\mathrm{Var}}}
\newcommand{\Cl}{{\mathrm{Cl}}}
\renewcommand{\middle}{{\mathrm{mid}}}
\newcommand{\lay}{{\mathrm{lay}}}
\newcommand{\up}{{\mathrm{up}}}
\newcommand{\low}{{\mathrm{low}}}
\newcommand{\bc}{{\mathrm{bc}}}
\newcommand{\Dom}{{\mathrm{Dom}}}
\newcommand{\clone}{{\mathrm{clone}}}
\newcommand{\twin}{{\mathrm{twin}}}
\newcommand{\opp}{{\mathrm{opposite}}}
\newcommand{\prob}[1]{\mathrm{Pr}[#1]}

\newcommand{\heading}[1]{\medskip\noindent{\bf #1.\ }}%

\newcommand{\barE}{{\ensuremath{\bar{E}}}}
\newcommand{\barG}{{\ensuremath{\bar{G}}}}
\newcommand{\bibarG}{{\ensuremath{\hat{G}}}}

\newcommand{\Pp}{{\ensuremath{\mathcal{P}}}}

\def\id{{\rm id}}
\def\prj{{\rm prj}}
\def\dist{{\rm dist}}

\begin{document}
\title{On the fine-grained complexity of rainbow coloring\thanks{Work partially supported by the National Science Centre of Poland, grant number 2013/09/B/ST6/03136 (\L{}.K., A.S.), and by the Emil Aaltonen Foundation (J.L.). } 
}

\date{}

\author{\L ukasz Kowalik\thanks{University of Warsaw, Poland} \and Juho Lauri\thanks{Tampere University of Technology, Finland} \and Arkadiusz Soca\l a\thanks{University of Warsaw, Poland}}

\maketitle

\begin{abstract}
The \probRC problem asks whether the edges of a given graph can be colored in $k$ colors so that every pair of vertices is connected by a rainbow path, i.e., a path with all edges of different colors. 
Our main result states that for any $k\ge 2$, there is no algorithm for \probRC running in time $2^{o(n^{3/2})}$, unless ETH fails.
Motivated by this negative result we consider two parameterized variants of the problem.
In \probPartRC problem, introduced by Chakraborty \emph{et al.}~[STACS 2009, J. Comb. Opt. 2009], we are additionally given a set $S$ of pairs of vertices and we ask if there is a coloring in which all the pairs in $S$ are connected by rainbow paths. 
We show that \probPartRC is FPT when parameterized by $|S|$.
We also study \probMaxRC problem, where we are additionally given an integer $q$ and we ask if there is a coloring in which at least $q$ anti-edges are connected by rainbow paths. 
We show that the problem is FPT when parameterized by $q$ and has a kernel of size $O(q)$ for every $k\ge 2$ (thus proving that the problem is FPT), extending the result of Ananth \emph{et al.}~[FSTTCS 2011].
\end{abstract}

\section{Introduction}

The \probRC problem asks whether the edges of a given graph can be colored in $k$ colors so that every pair of vertices is connected by a rainbow path, i.e., a path with all edges of different colors.
Minimum such $k$, called the {\em rainbow connection number} can be viewed as yet another measure of graph connectivity.
The concept of rainbow coloring was introduced by Chartrand, Johns, McKeon, and Zhang~\cite{Chartrand2008:mb} in 2008, while also featured in an earlier book of Chartrand and Zhang~\cite{Chartrand-book}.
Chakraborty, Fischer, Matsliah, and Yuster~\cite{Chakraborty:JCombOpt} describe an interesting application of rainbow coloring in telecommunications.
The problem is intensively studied from the combinatorial perspective, with over 100 papers published by now (see the survey of Li, Shi, and Sun~\cite{Li2012-survey} for an overview).
However, computational complexity of the problem seems less explored.
It was conjectured by Caro, Lev, Roditty, Tuza, and Yuster~\cite{Caro2008:ejc} that the \probRC problem is NP-complete for $k=2$.
This conjecture was confirmed by Chakraborty \emph{et al.}~\cite{Chakraborty:JCombOpt}.
Ananth, Nasre, and Sarpatwar~\cite{Ananth:fsttcs11} noticed that the proof of Chakraborty \emph{et al.} in fact proves NP-completeness for every even $k>1$, and complemented this by showing NP-completeness of the odd cases as well. 
An alternative hardness proof for every $k > 1$ was provided by Le and Tuza~\cite{Le:tech}. 
For complexity results on restricted graph classes, see e.g.,~\cite{Chandran2012:coc,Chandran:fsttcs13,Chandran2015,Eiben:iwoca15}.

For many NP-complete graph problems there are algorithms running in time $2^{O(n)}$ for an $n$-vertex graph.
This is obviously the case for problems asking for a set of vertices, like {\sc Clique} or {\sc Vertex Cover}, or more generally, for problems which admit polynomially (or even subexponentially) checkable  $O(n)$-bit certificates.
However, there are $2^{O(n)}$-time algorithms also for some problems for which such certificates are not known, including e.g., {\sc Hamiltonicity}~\cite{HeldKarp62} and {\sc Vertex Coloring}~\cite{Lawler76}.
Unfortunately it seems that the best known worst-case running time bound for \probRC is $k^m 2^n n^{O(1)}$, where $m$ is the number of edges, which is obtained by checking each of the $k^m$ colorings by a simple $2^n n^{O(1)}$-time dynamic programming algorithm~\cite{Uchizawa-cocoon11}.
Even in the simplest variant of just two colors, i.e., $k=2$, this algorithm takes $2^{O(n^2)}$ time if the input graph is dense.
It raises a natural question: is this problem really much harder than, say, {\sc Hamiltonicity}, or have we just not found the right approach yet?
Questions of this kind have received considerable attention recently.
In particular, it was shown that unless the Exponential Time Hypothesis fails, there is no algorithm running in time $2^{o(n\log n)}$ for {\sc Channel Assignment}~\cite{Socala15}, {\sc Subgraph Homomorphism}, and {\sc Subgraph Isomorphism}~\cite{CyganFGKMPS16}. 
Let us recall the precise statement of the Exponential Time Hypothesis (ETH).

\begin{conjecture}[Exponential Time Hypothesis~\cite{eth}]
There exists a constant $c > 0$, such that 
there is no algorithm solving \probTSAT in time $O^*(2^{cn})$.
\end{conjecture}

Note that some kind of a complexity assumption, like ETH, is hard to avoid when we prove exponential lower bounds, unless one aims at proving ${\rm P} \ne {\rm NP}$.

\heading{Main Result}
Our main result states that for any $k \ge 2$ there is no algorithm for \probRC running in time $2^{o(n^{3/2})}$, unless the Exponential Time Hypothesis fails.
To our best knowledge this is the first NP-complete graph problem for which the existence of a $2^{o(n^{1+\epsilon})}$-time algorithm is excluded (under reasonable complexity assumptions), for an $\epsilon>0$.

\heading{Remaining Lower Bounds}
The proof of our main result implies a few corollaries, which may be of independent interest.
First, we show that ETH implies that for any $k\ge 2$, \probRC has no algorithm running in time $2^{o(m/\log m)}$, where $m$ is the number of edges. 
This shows that the best known algorithm, running in time $2^{m\log k+n}n^{O(1)}$, is not far from being optimal if we consider the problem as parameterized by the number of edges.
Second, we study a generalized problem, called \probPartRC, introduced by Chakraborty \emph{et al.}~\cite{Chakraborty:JCombOpt} as a natural intermediate step in reductions from 3-SAT to \probRC.
In \probPartRC, we are given a connected graph $G$, and a set of pairs of vertices $S\subseteq {V(G) \choose 2}$.
Elements of $S$ are called {\em requests}.
For a given coloring of $E(G)$ we say that a request $\{u,v\}$ is {\em satisfied} if $u$ and $v$ are connected by a rainbow path.
The goal in \probPartRC is to determine whether there is a $k$-coloring of $E(G)$ such that every pair in $S$ is satisfied.
Our main result implies that \probPartRC admits no algorithm running in time $2^{o(n^{3/2})}$, under ETH.
Moreover, we show that ETH implies that this problem admits neither $2^{o(m)}$ nor $2^{o(|S|)}$ running time.
An interesting feature here is that for $k=2$ these bounds are {\em tight} up to a polynomial factor (a $2^mn^{O(1)}$ algorithm is immediate, and a $2^{|S|}n^{O(1)}$-time algorithm is discussed in the next paragraph).

\heading{New Algorithms}
In the context of the hardness results mentioned above it is natural to ask for FPT algorithms for \probPartRC.
We show that for every fixed $k$, \probPartRC parameterized by $|S|$ is FPT: we show an algorithm running in time $|S|^{O(|S|)}n^{O(1)}$.
For the 2 color case we are able to show a different, faster algorithm running in time $2^{|S|}n^{O(1)}$, which is tight up to a polynomial factor.
We also study the \probMaxRC problem, introduced by Ananth, Nasre, and Sarpatwar~\cite{Ananth:fsttcs11}.
Intuitively, the idea is to parameterize the problem by the number of pairs to satisfy.
However, all pairs of adjacent vertices are trivially satisfied by any edge-coloring.
Hence, we parameterize by the number of anti-edges to satisfy.
More formally, in \probMaxRC we are given a graph $G=(V,E)$, an integer $q$, and asked whether there is a coloring of $E$ that satisfies at least $q$ anti-edges.
First, we show that the maximization version of the problem (find maximum such $q$) admits a constant factor approximation algorithm for every fixed value of $k$. 
Second, we show that \probMaxRC is FPT for every $k\ge 2$, which generalizes the result of Ananth \emph{et al.}~\cite{Ananth:fsttcs11} who showed this claim for the $k=2$ case.
Our algorithm runs in time $2^{q\log q}n^{O(1)}$ for any $k$, which is faster than the algorithm of Ananth \emph{et al.} for 2 colors.
For 2 colors we give an even faster algorithm, running in time $8^qn^{O(1)}$.
We also show that the problem admits a kernel size $O(q)$, i.e., that there is a polynomial-time algorithm that returns an equivalent instance with $O(q)$ vertices.
(For more background on kernelization see e.g.,~\cite{fptbook}.)
Before, this was known only for $k=2$ (due to Ananth \emph{et al.}~\cite{Ananth:fsttcs11}).

\subsection{Notation}
For standard graph-theoretic notions, we refer the reader to~\cite{Diestel-book}.
All graphs we consider in this paper are simple and undirected.
We denote $\Delta_1(G)=\max\{\Delta(G),1\}$.

A \emph{rainbow walk} is a walk with all edges of different colors.
By $\barE$ we denote the set of anti-edges, i.e., $\barE = {V \choose 2} \setminus E$.
When $G=(V,E)$ is a graph then $\barG=(V,\barE)$ is its {\em complement graph}.
By $x^{\underline{k}}$ we denote the falling factorial, i.e., $x^{\underline{k}} = x(x-1)\cdots(x-k+1)$.

If $I$ and $J$ are instances of decision problems $P$ and $R$, respectively, then we say that $I$ and $J$ are {\em equivalent}, when either both $I$ and $J$ are YES-instances or both are NO-instances.

\subsection{Organization of the paper}
In Section~\ref{sec:hardness} we present our hardness results.
Sections~\ref{sec:alg-subset} and~\ref{sec:alg:max} contain our algorithms for \probPartRC and \probMaxRC, respectively.
Finally, in Section~\ref{sec:further} we discuss some directions of further work.

\section{Hardness of rainbow coloring}
\label{sec:hardness}

The main goal of this section is to show that for any $k\ge 2$ \probRC does not admit an algorithm running in time $2^{o(n^{3/2})}$, unless the Exponential Time Hypothesis fails.
Let us give a high-level overview of our proof.
A natural idea would be to begin with a \probTSAT formula $\phi$ with $n$ variables and then transform it in time $2^{o(n)}$ to an equivalent instance $G=(V,E)$ of \probRC with $O(n^{2/3})$ vertices.
Then indeed a $2^{o(|V|^{3/2})}$-time algorithm that solves \probRTwoC can be used to decide \probTSAT in time $2^{o(n)}$.
Note that in a typical NP-hardness reduction, we observe some polynomial blow-up of the instance size.
For example, one can verify that in the reduction of Chakraborty \emph{et al.}~\cite{Chakraborty:JCombOpt}, the initial \probTSAT formula with $n$ variables and $m$ clauses is transformed into a graph with $\Theta(n^4+m^4)$ vertices and edges. 
In our case, instead of a blow-up we aim at {\em compression}: the number of vertices needs to be much smaller than the number of variables in the input formula $\phi$.
As usual in reductions, variables and clauses in $\phi$ are going to correspond to some structures in $G$, called gadgets.
The compression requirement means that our gadgets need to share vertices.
The more clauses we have the harder this task is. 
For that reason, we apply the following well-known Sparsification Lemma, which allows for assuming that the number of clauses is $O(n)$.

\begin{lemma}[Sparsification Lemma~\cite{seth}]
\label{lem_sparsification}
For each $\eps > 0$ there exist a constants $c_{\eps}$,
such that any \probTSAT formula $\varphi$ with $n$ variables
can be expressed as $\varphi = \vee_{i=1}^t \psi_i$, where $t \le 2^{\eps n}$ 
and each $\psi_i$ is a \probTSAT formula with the same variable set
as $\varphi$, but contains at most $c_{\eps}n$ clauses.
Moreover, this disjunction can be computed in time $O^*(2^{\eps n})$.
\end{lemma}

Note that by using the Sparsification Lemma we modify our general plan a bit: instead of creating one equivalent instance, we are going to create $2^{\eps n}$ instances (for arbitrarily small $\epsilon$), each with $O(n^{2/3})$ vertices. The following lemma further simplifies the instance. 

\begin{lemma}[\cite{Tovey84}]
\label{lem_transformation}
Given a \probTSAT formula $\varphi$ with $m$ clauses one can transform
it in polynomial time into a formula $\varphi'$ with $O(m)$ 
variables and $O(m)$ clauses, such that $\varphi'$ is
satisfiable iff $\varphi'$ is satisfiable, and
moreover each clause of $\varphi'$ contains exactly three different
variables and each variable occurs in at most $4$ clauses
of $\varphi'$.
\end{lemma}

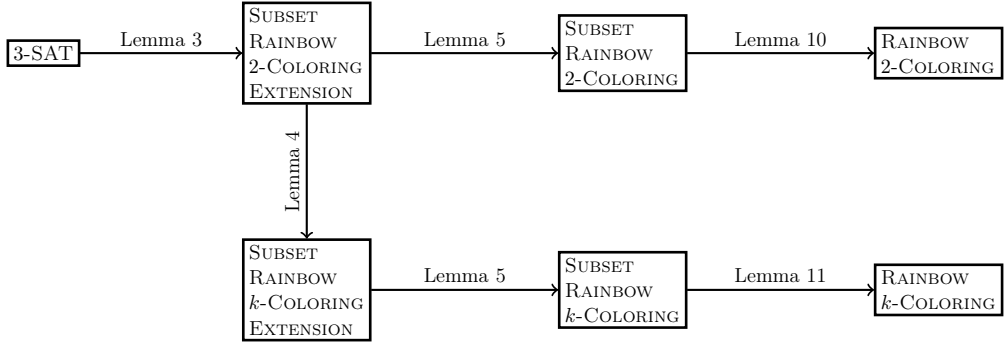
\begin{figure}
\begin{center}
 \begin{tikzpicture}[scale=0.7, every node/.style={scale=0.7}]
\tikzset{prob/.style={draw, rectangle, line width=1pt, inner sep = 3pt}}

\node [prob] (sat) at (0,0) {\probTSAT};
\node [text width=2.2cm,prob] (sr2cext) at (5,0) {\probPartRTwoCExt};
\node [text width=2.2cm,prob] (srkcext) at (5,-4.5) {\probPartRCExt};
\node [text width=2.2cm,prob] (sr2c) at (11,0) {\probPartRTwoC};
\node [text width=2.2cm,prob] (srkc) at (11,-4.5) {\probPartRC};
\node [text width=2.2cm,prob] (r2c) at (17,0) {\probRTwoC};
\node [text width=2.2cm,prob] (rkc) at (17,-4.5) {\probRC};

\draw [->,thick] (sat) -- node [above] {Lemma~\ref{lem:from-TSAT-to-PartRTwoCExt}} (sr2cext);
\draw [->,thick] (sr2cext) -- node [sloped,above] {Lemma~\ref{lem:from-PartRTwoCExt-to-PartRCExt}} (srkcext);
\draw [->,thick] (sr2cext) -- node [sloped,above] {Lemma~\ref{lem:from-PartRCExt-to-PartRC}} (sr2c);
\draw [->,thick] (srkcext) -- node [sloped,above] {Lemma~\ref{lem:from-PartRCExt-to-PartRC}} (srkc);
\draw [->,thick] (sr2c) -- node [sloped,above] {Lemma~\ref{lem:from-PartRTwoC-to-RTwoC}} (r2c);
\draw [->,thick] (srkc) -- node [sloped,above] {Lemma~\ref{lem:from-PartRC-to-RC}} (rkc);
\end{tikzpicture}
\caption{\label{fig:roadmap}The road map of our reductions.}
\end{center}
\end{figure}

Now our goal is to transform a \probTSAT formula $\phi$ with $n$ variables such that every variable occurs in at most 4 clauses, to a graph with $O(n^{2/3})$ vertices --- an equivalent instance of \probRC. We do it in four steps (see Fig~\ref{fig:roadmap}).

In the first step we transform $\phi$ to an instance $I=(G,S,c_0)$ of \probPartRTwoCExt, which is a generalization of \probPartRTwoC, where $c_0$, called a {\em precoloring}, is a partial coloring of the edges of $G$ into two colors and the goal is to determine if there is an edge-coloring of $E(G)$ which extends $c_0$ and such that all pairs of $S$ are satisfied.
The first step is crucial, because here the compression takes place: $|V(G)|=O(n^{2/3})$ and $E(G)=O(n)$.
The major challenge in the construction is avoiding interference between gadgets that share a vertex: to this end we define various conflict graphs and we show that they can be vertex-colored in a few colors. This reduction is described in Section~\ref{sec:TSAT->PartRTwoCExt}.

In the second step (Section~\ref{sec:PartRTwoCExt->PartRCExt}) we reduce \probPartRTwoCExt to \probPartRCExt, for every $k\ge 3$.
In the third step (Section~\ref{sec:PartRCExt->PartRC}) an instance of \probPartRCExt is transformed to an instance of \probPartRC, for every $k\ge 3$. 
The number of the vertices in the resulting instance does not increase more than by a constant factor.
These steps are rather standard, though some technicalities appear because we need to guarantee additional properties of the output instance, which are needed by the reduction in the fourth step. 

The last step (Section~\ref{sec:PartRC->RC}), where we reduce an instance $(G=(V,E),S)$ of \probPartRC to an instance $G'$ of \probRC, is another challenge.
We would like to get rid of the set of requests somehow.
For simplicity, let us focus on the $k=2$ case now.
Here, the natural idea, used actually by Chakraborty \emph{et al.}~\cite{Chakraborty:JCombOpt} is to create, for every $\{u,v\}\not\in S$, a path $(u,x_{uv},v)$ through a new vertex $x_{uv}$. Such a path cannot help any of the requests from $S$ to get satisfied, and by coloring it into two different colors we can satisfy $\{u,v\}$. 
Unfortunately, in our case we cannot afford for creating a new vertex for every such $\{u,v\}$, because that would result in a quadratic blow up in the number of vertices.
However, one can observe that for any biclique (a complete bipartite subgraph) in the graph $(V,{V \choose 2} \setminus S)$ it is sufficient to use just one such vertex $x$ (connected to all the vertices of the biclique).
By applying a result of Jukna~\cite{Jukna-on-set-intersection-representations} we can show that in our specific instance of \probPartRTwoC which results from a \probTSAT formula, the number of bicliques needed to cover all the pairs in ${V \choose 2} \setminus S$ is small enough. 
We show a $2^{|V(G)|}|V(G)|^{O(1)}$-time algorithm to find such a cover.
Although this algorithm does not seem fast, in our case $|V(G)|=O(n^{2/3})$, so this complexity is {\em subexponential} in the number of variables of the input formula, which is enough for our goal. 
The case of $k\ge 3$ is similar, i.e., we also use the biclique cover. 
However, the details are much more technical because for each biclique we need to introduce a much more complex gadget.

\subsection{From \probTSAT to \probPartRTwoCExt}
\label{sec:TSAT->PartRTwoCExt}

Let \probPartRCExt be a generalization of \probPartRC,
where $c_0$ is a partial $k$-coloring of the edges of $G$
and the goal is to determine if there is an edge-coloring of $E(G)$
which extends $c_0$ and such that all pairs of $S$ are satisfied.
For an instance $I=(G,S,c_0)$ of \probPartRCExt (for any $k\ge 2$), let us define a {\em precoloring conflict graph} $CG_I$.
Its vertex set is the set of colored edges, i.e., $V(CG_I)=\Dom(c_0)$.
Two different colored edges $e_1$ and $e_2$ are {\em adjacent} in $CG_I$ when they are incident in $G$ or there is a pair of endpoints $u\in e_1$ and $v\in e_2$ such that $uv \in E(G)\cup S$.

In what follows the reduction in Lemma~\ref{lem:from-TSAT-to-PartRTwoCExt} is going to be pipelined with three further reductions going through \probPartRCExt and \probPartRC to \probRC.
In these three reductions we need to keep the instance small. 
To this end, the instance of \probPartRTwoCExt resulting in Lemma~\ref{lem:from-TSAT-to-PartRTwoCExt} has to satisfy some additional properties, which are formulated in the claim of Lemma~\ref{lem:from-TSAT-to-PartRTwoCExt}.
Their role will become more clear later on.

\begin{lemma} \label{lem:from-TSAT-to-PartRTwoCExt}
Given a \probTSAT formula $\varphi$ with $n$ variables such that each clause of $\varphi$ contains exactly three
variables and each variable occurs in at most four clauses, one can construct in polynomial time an equivalent instance $(G,S,c_0)$ of \probPartRTwoCExt such that $G$ has $O(n^{2/3})$ vertices and $O(n)$ edges.
Moreover, $\Delta(G) = O(n^{1/3})$, $\Delta(V(G),S)=O(n^{1/3})$, $|\Dom(c_0)| = O(n^{2/3})$ and along with the instance $I=(G,S,c_0)$ the algorithm constructs a proper vertex $4$-coloring of $(V(G),E \cup S)$ (so also of $(V(G), S)$) and a proper vertex $O(n^{1/3})$-coloring of the precoloring conflict graph $CG_I$.
\end{lemma}

\begin{proof}
 Let $m$ denote the number of clauses in $\varphi$. Observe that $m \le \frac{4}3 n$.
 Let $\Var$ and $\Cl$ denote the sets of variables and clauses of $\varphi$.
 For more clarity, the two colors of the partial coloring $c_0$ will be called $T$ and $F$.
 
 Let us describe the graph $G$ along with a set of anti-edges $S$.
 Graph $G$ consists of two parts: the variable part and the clause part.
 The intuition is that in any 2-edge coloring of $G$ that extends $c_0$ and satisfies all pairs in $S$
 \begin{itemize}
  \item colors of the edges in the variable part represent an assignment of the variables of $\varphi$,
  \item colors of the edges in the clause part represent a choice of literals that satisfy all the clauses, and
  \item colors of the edges between the two parts make the values of the literals from the clause part consistent with the assignment represented by the variable part.
 \end{itemize}
 
 \heading{The variable part}
 The vertices of the variable part consist of the {\em middle set} $M$ and $\ceil{n^{1/3}}$ layers $L_1 \cup L_2 \cdots \cup L_{\ceil{n^{1/3}}}$.
 The middle set $M$ consists of vertices $m_i$ for each $i=1,\ldots,\ceil{n^{2/3}}+9$.
 For every $i=1,\ldots,\ceil{n^{1/3}}$ the layer $L_i$ consists of two parts: upper $L_i^{\uparrow} = \{u_{i,j} \ :\ j=1,\ldots, \ceil{n^{1/3}}+3\}$ and lower $L_i^{\downarrow} = \{l_{i,j} \ :\ j=1,\ldots, \ceil{n^{1/3}}+3\}$.
 
 We are going to define four functions: $\middle : \Var \rightarrow M$, $\lay, \up, \low : \Var \rightarrow [\ceil{n^{1/3}}]$.
 Then, for every variable $x\in \Var$ we add two edges $u_{\lay(x),\up(x)} \middle(x)$ and $\middle(x) l_{\lay(x),\low(x)}$.
 Moreover, we add the pair $p_x=\{u_{\lay(x),\up(x)}, l_{\lay(x),\low(x)}\}$ to $S$.
 In other words, $x$ corresponds to the 2-path $u_{\lay(x),\up(x)} \middle(x) l_{\lay(x),\low(x)}$.
 Now we describe a careful construction of the four functions, that guarantee several useful properties (for example edge-disjointness of paths corresponding to different variables).
 
 Let us define the {\em variable conflict graph} $G_V=(\Var, E_{G_V})$, where for two variables $x,y\in \Var$ we have $xy$ are adjacent iff they both occur in the same clause.
 Since every variable occurs in at most 4 clauses, $\Delta(G_V) \le 8$.
 It follows that there is a proper vertex 9-coloring $\alpha:Var \rightarrow [9]$ of $G_v$, and it can be found by a simple linear time algorithm.
 Next, each of the 9 color classes $\alpha^{-1}(i)$ is partitioned into $\ceil{|\alpha^{-1}(i)| / \ceil{n^{1/3}}}$ disjoint groups, each of size at most $\ceil{n^{1/3}}$.
 It follows that the total number $n_g$ of groups is at most $\ceil{n^{2/3}}+9$. 
 Let us number the groups arbitrarily from 1 to $n_g$ and for every variable $x\in \Var$, let $g(x)$ be the number of the group that contains $x$. 
 Then we define $\middle(x) = m_{g(x)}$. Since any group contains only vertices of the same color we can state the following property:
 
 \begin{enumerate}
 	\item [$(P_1)$] If variables $x$ and $y$ occur in the same clause then $\middle(x) \ne \middle(y)$.
 \end{enumerate}

 Now, for every variable $x$ we define its layer, i.e., the value of the function $\lay(x)$.
 Recall that for every $i=1,\ldots, \ceil{n^{2/3}}+9$ the $i$-th group $\middle^{-1}(m_i)$ contains at most $\ceil{n^{1/3}}$ variables.
 Inside each group, number the variables arbitrarily and let $\lay(x)$ be the number of variable $x$ in its group, $\lay(x)\in [n^{1/3}]$.
 This implies another important property.
 
 \begin{enumerate}
  \item [$(P_2)$] If variables $x$ and $y$ belong to the same layer then $\middle(x) \ne \middle(y)$.
 \end{enumerate}

 Observe that every layer gets assigned at most $\ceil{n^{2/3}}+9$ variables.
 For every layer $L_i$ pick any injective function $h_i:\lay^{-1}(i) \rightarrow [\ceil{n^{1/3}}+3]^2$.
 Then, for every variable $x\in \Var$ we put $(\up(x),\low(x)) = h_{\lay(x)} (x)$.
 Note that by $(P_2)$ we have the following.
 
 \begin{enumerate}
  \item [$(P_3)$] For every variable $x$ there is exactly one 2-path in $G$ connecting $p_x$, namely $(u_{\lay(x),\up(x)},\middle(x),l_{\lay(x),\low(x)})$.

  \item [$(P_4)$] For every pair of variables $x,y$ the two unique paths connecting $p_x$ and $p_y$ are edge-disjoint.
 \end{enumerate}

 Although we are going to add more edges and vertices to $G$,
 none of these edges has any endpoint in $\bigcup_i (U_i \cup L_i) $,
 so $P_3$ will stay satisfied.
 
 \heading{The clause part}
 The vertices of the clause part are partitioned into $O(m^{1/3})$ clusters.
 Similarly as in the case of variables, each clause is going to correspond to a pair of vertices in the same cluster.
 Again, the assignment of clauses to clusters has to be done carefully.
 To this end we introduce the {\em clause conflict graph} $G_C=(\Cl,E_{G_C})$. 
 Two different clauses $C_1$ and $C_2$ are adjacent in $G_C$ if $C_1$ contains a variable $x_1$ and $C_2$ contains a variable $x_2$ such that $\middle(x_1)=\middle(x_2)$.
 Fix a variable $x_1$. 
 Since $|\middle^{-1}(\middle(x_1))| \le \ceil{n^{1/3}}$, there are at most $\ceil{n^{1/3}}$ variables $x_2$ such that $\middle(x_1)=\middle(x_2)$.
 Since every clause contains $3$ variables, and each of them is in at most $4$ clauses, $\Delta(G_C)\le 12\ceil{n^{1/3}}$.
 It follows that in polynomial time we can find a proper coloring $\beta$ of the vertices of $G_C$ into at most $12\ceil{n^{1/3}}+1$ colors.
 Moreover, if for any color $j$ its color class $\beta^{-1}(j)$ is larger than $\ceil{n^{2/3}}$ we partition it into $\ceil{|\beta^{-1}(j)|/\ceil{n^{2/3}}}$ new colors.
 Clearly, in total we produce at most $\frac{4}{3}\ceil{n^{1/3}}$ new colors in this way.
 Hence, in what follows we assume that each color class of $\beta$ is of size at most $\ceil{n^{2/3}}$, and the total number of colors $s\le 13\ceil{n^{1/3}}+1$.
 In what follows we construct $s$ clusters $Q_1,\ldots,Q_s$.
 Every clause $C\in \Cl$ is going to correspond to a pair of vertices in the cluster $Q_{\beta(C)}$.

Fix $i=1,\ldots,s$. Let us describe the subgraph induced by cluster $Q_i$.
Define {\em cluster conflict graph} $G_i=(\beta^{-1}(i),E_{G_i})$. 
Two different clauses $C_1, C_2 \in \beta^{-1}(i)$ are adjacent in $G_i$ if there are three variables $x_1$, $x_2$, and $x_3$ such that 
 \begin{enumerate}[$(i)$]
  \item $C_1$ contains $x_1$,
  \item $C_2$ contains $x_2$,
  \item $(\lay(x_1),\up(x_1))=(\lay(x_3),\up(x_3))$, and
  \item $\middle(x_2)=\middle(x_3)$.
 \end{enumerate}
Fix a variable $x_1$ which appears in a clause $C_1\in\beta^{-1}(i)$. 
By our construction, there are at most $\ceil{n^{1/3}}+2$ other variables $x_3$ that map to the same pair as $x_1$ by functions $\lay$ and $\up$.
For each such $x_3$ there are at most $\ceil{n^{1/3}}$ variables $x_2$ such that $\middle(x_2)=\middle(x_3)$; however, at most one of these variables belongs to a clause $C_2$ from the same cluster $\beta^{-1}(i)$, by the definition of the coloring $\beta$.
It follows that $\Delta(G_i)\le 12(\ceil{n^{1/3}}+2)$.
Hence in polynomial time we can find a proper coloring $\gamma_i$ of the vertices of $G_i$ into at most $12(\ceil{n^{1/3}}+2)+1$ colors.
Similarly as in the case of the coloring $\beta$, we can assume that each of the color classes of $\gamma_i$ has at most $\ceil{n^{1/3}}$ clauses, at the expense of at most $\ceil{n^{1/3}}$ additional colors.
It follows that we can construct in polynomial time a function $g:\Cl \rightarrow [\ceil{n^{1/3}}]$ such that for every cluster $i=1,\ldots,s$ and for every color class $S$ of $\gamma_i$ $g$ is injective on $S$.
Let $n_i \le 13\ceil{n^{1/3}}+25$ be the number of colors used by $\gamma_i$.
For notational convenience, let us define a function $\gamma:\Cl\rightarrow [\max_i{n_i}]$ such that for any clause $C$ we have $\gamma(C)=\gamma_{\beta(C)}(C)$.

We are ready to define the vertices and edges of $Q_i$.
It is a union of three disjoint vertex sets $A_i$, $B_i$, and $C_i$.
We have $A_i=\{a_{i,j} \ :\ j = 1,\ldots,\ceil{n^{1/3}}\}$, $B_i=\{b_{i,j}^k \ :\ j = 1,\ldots,n_i, k=1,2,3\}$, and $C_i=\{c_{i,j} \ :\ j = 1,\ldots,n_i\}$.
For every $j = 1,\ldots,n_i$ and for every $k=1,2,3$ we add edge $c_{i,j} b_{i,j}^k$ to $G$, and we color it by $c_0$ to color $F$.
(These are the only edges pre-colored in the whole graph $G$.)
For every clause $C\in \beta^{-1}(i)$ we do the following. 
For each $k=1,2,3$, add the edge $(a_{i,g(C)},b_{i,\gamma(C)}^k)$ to $G$.
Finally, add the pair $\{a_{i,g(C)},c_{i,\gamma(C)}\}$ to $S$. 
Clearly, the following holds:

 \begin{enumerate}
  \item [$(P_5)$] Let $C$ be any clause. Let $i=\beta(C)$ and let $j = g(C)$. 
  Then there are exactly three 2-paths between $a_{\beta(C),g(C)}$ and $c_{\beta(C),\gamma(C)}$, each going through $b_{\beta(C),\gamma(C)}^k$ for $k=1,2,3$.
 \end{enumerate}
 
The description of clusters is now finished.  

\heading{Connections between the two parts}
Consider a clause $C=\{\ell_1,\ell_2,\ell_3\}$ and its $k$-th literal $\ell_k$ for each $k=1,2,3$.
Then for some variable $x$ we have $\ell_k=x$ or $\ell_k=\bar{x}$.
We add the edge $b_{\beta(C),\gamma(C)}^k \middle(x)$ and we add the pair $\{\middle(x), a_{\beta(C),g(C)}\}$ to $S$.
If $\ell_k=x$, we also add the pair $\{b_{\beta(C),\gamma(C)}^k, u_{\lay(x),\up(x)}\}$ to $S$;
otherwise we add the pair $\{b_{\beta(C),\gamma(C)}^k, l_{\lay(x),\low(x)}\}$ to $S$.
We claim the following.

 \begin{enumerate}
  \item [$(P_6)$] Every edge between the two parts was added exactly once, i.e., for every edge $uv$ such that $u$ is in the clause part and $v$ is in the variable part, there is exactly one clause $C$ and exactly one literal $\ell_k \in C$ such that $u=b_{\beta(C),\gamma(C)}^k$ and $v=\middle(x)$, where $x$ is the variable in $\ell_k$.
 \end{enumerate}

 Indeed, assume for a contradiction that there is a clause $C_1$ with its $k_1$-th literal containing $x_1$ and a clause $C_2$ with its $k_2$-th literal containing $x_2$ such that 
 $b_{\beta(C_1),\gamma(C_1)}^{k_1}=b_{\beta(C_2),\gamma(C_2)}^{k_2}$ and $\middle(x_1)=\middle(x_2)$.
 Then $C_1\ne C_2$ by $(P_1)$. Since $\middle(x_1)=\middle(x_2)$, $C_1$ and $C_2$ are adjacent in the clause conflict graph $G_C$.
 It follows that $\beta(C_1)\ne\beta(C_2)$, so two different clusters share a vertex, a contradiction.

This finishes the description of the instance $(G,S,c_0)$. (See Fig.~\ref{fig:main-red}.)

\begin{figure}[t]
	\begin{tikzpicture}[scale=0.6]
\tikzstyle{every node}=[circle, draw, color=black, fill=black, inner sep=0pt, minimum width=10pt]

\draw [rounded corners, thick] (-9,5.5) rectangle (-6,3.5);
\draw [rounded corners, thick] (-5.5,5.5) rectangle (-2.5,3.5);
\draw [rounded corners, thick] (-2,5.5) rectangle (1,3.5);

\draw [rounded corners, thick] (-8.5,2.5) rectangle (0.5,0.5);

\draw [rounded corners, thick] (-9,-0.5) rectangle (-6,-2.5);
\draw [rounded corners, thick] (-5.5,-0.5) rectangle (-2.5,-2.5);
\draw [rounded corners, thick] (-2,-0.5) rectangle (1,-2.5);

\draw [rounded corners, thick] (3,2.5) rectangle (12,0.5);
\draw [rounded corners, thick] (3.5,5.5) rectangle (11.5,3.5);
\draw [rounded corners, thick] (3.5,-0.5) rectangle (11.5,-2.5);

\node [draw=none, fill=none] at (-4,-4.5) {
  Variable Gadget};
\node [draw=none, fill=none] at (7.5,-4.5) {
  Clause Gadget (one of $O(n^{1/3})$ clusters)};

\node [draw=none, fill=none] at (-7.5,-3) {
  $O(n^{1/3})$};
\node [draw=none, fill=none] at (-4,-3) {
  $O(n^{1/3})$};
\node [draw=none, fill=none] at (-0.5,-3) {
  $O(n^{1/3})$};
\node [draw=none, fill=none] at (-7.5,6) {
  $O(n^{1/3})$};
\node [draw=none, fill=none] at (-4,6) {
  $O(n^{1/3})$};
\node [draw=none, fill=none] at (-0.5,6) {
  $O(n^{1/3})$};
\node [draw=none, fill=none] at (7.5,6) {
  $O(n^{1/3})$};
\node [draw=none, fill=none] at (7.5,-3) {
  $O(n^{1/3})$};
\node [draw=none, fill=none] at (13.5,1.5) {
  $O(n^{1/3})$};
  
\node [draw=none, fill=none] at (-10,1.5) {
  $O(n^{2/3})$};

\node at (-1.5,4.5) {};
\node at (-7.5,4.5) {};
\node at (-6.5,4.5) {};
\node at (-8,1.5) {};
\node at (-7,1.5) {};
\node at (-6,1.5) {};
\node at (-5,1.5) {};
\node (v32) at (-4,1.5) {};
\node at (-3,1.5) {};
\node at (-2,1.5) {};
\node at (-1,1.5) {};
\node at (0,1.5) {};
\node at (-5,4.5) {};
\node at (-4,4.5) {};
\node at (-3,4.5) {};
\node (v31) at (-8.5,4.5) {};
\node at (-0.5,4.5) {};
\node at (0.5,4.5) {};
\node at (-8.5,-1.5) {};
\node at (-7.5,-1.5) {};
\node at (-0.5,-1.5) {};
\node at (-5,-1.5) {};
\node at (-4,-1.5) {};
\node at (-3,-1.5) {};
\node at (-1.5,-1.5) {};
\node (v33) at (-6.5,-1.5) {};
\node at (0.5,-1.5) {};
\node (v17) at (3.5,1.5) {};
\node (v18) at (4.5,1.5) {};
\node (v19) at (5.5,1.5) {};
\node (v20) at (6.5,1.5) {};
\node (v22) at (7.5,1.5) {};
\node (v23) at (8.5,1.5) {};
\node (v25) at (9.5,1.5) {};
\node (v26) at (10.5,1.5) {};
\node (v27) at (11.5,1.5) {};
\node (v28) at (4.5,4.5) {};
\node (v29) at (7.5,4.5) {};
\node (v30) at (10.5,4.5) {};
\node (v16) at (4.5,-1.5) {};
\node (v21) at (7.5,-1.5) {};
\node (v24) at (10.5,-1.5) {};
\draw  [very thick, red] (v16) edge (v17);
\draw  [very thick, red] (v16) edge (v18);
\draw  [very thick, red] (v16) edge (v19);

\draw  [thick] (v17) edge (v29);
\draw  [thick] (v18) edge (v29);
\draw  [thick] (v19) edge (v29);

\draw  [thick] (v31) edge (v32);
\draw  [thick] (v32) edge (v33);

\draw  [thick, bend right] (v32) edge (v17);
\draw  [very thick, dashed] (v31) edge (v17);
\draw  [very thick, dashed] (v32) edge (v29);
\draw  [very thick, dashed, bend left] (v29) edge (v16);
\draw  [very thick, dashed, bend right] (v31) edge (v33);
\end{tikzpicture}
	\vspace{-80pt}
	\caption{\label{fig:main-red}A simplified view of the obtained instance.
		Edges (solid lines) and requests (dashed lines) representing one variable and one clause that contains this variable are presented on the
		picture.} 
\end{figure}
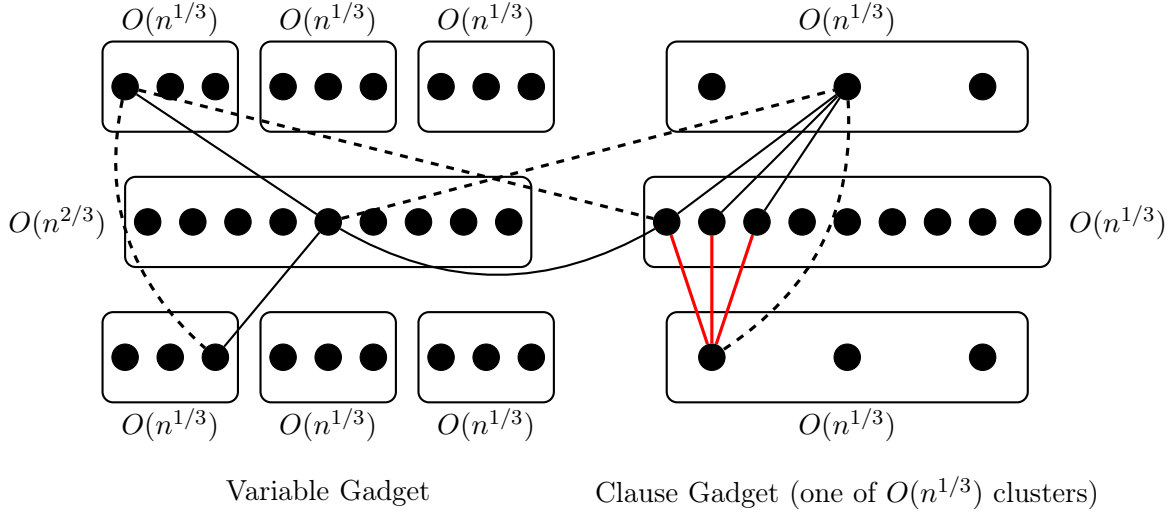

\heading{Size and time}
The construction clearly takes polynomial time.
In the variable part we have $|M| = \ceil{n^{2/3}}+9$ and each of the $\ceil{n^{1/3}}$ layers contains $2\ceil{n^{1/3}}+6$ vertices. 
It follows that the variable part contains $O(n^{2/3})$ vertices.
In the clause part we have $s = O(n^{1/3})$ clusters. 
For each $i=1,\ldots,s$, the cluster $Q_i$ has $4n_i + \ceil{n^{1/3}} = O(n^{1/3})$ vertices.
Hence the clause part also has $O(n^{2/3})$ vertices.

In the variable part there are two edges per variable, so $2n$ in total.
In the clause part there are $3n_i$ edges in $i$-th cluster, i.e., $O(n^{2/3})$ in total, and 3 edges per clause, i.e., $3m = O(n)$ in total.
Finally, for every clause we added 3 edges between the variable part and the clause part, so $3m=O(n)$ in total.
It follows that $|E(G)|=O(n)$.

\heading{Maximum degree of $G$}
Consider a vertex $u_{i,j} \in L_i^{\uparrow}$, for some $i=1,\ldots,\ceil{n^{1/3}}$ and $j=1,\ldots,\ceil{n^{1/3}}+3$.
The only edges incident to $u_{i,j}$ are those of the form $u_{i,j} \middle(x)$, for some variable $x$ such that $\lay(x)=i$ and $\up(x)=j$.
Since $h_i$ is injective, there are at most $\ceil{n^{1/3}}+3$ variables in $\lay^{-1}(i) \cap \up^{-1}(j)$.
Hence $\deg_G(u_{i,j}) \le \ceil{n^{1/3}}+3$.
Analogously, the same bound holds for vertices in $\bigcup_i L_i^{\downarrow}$.
Consider a vertex $m_j$ in $M$.
For every variable $x\in \middle^{-1}(m_j)$ vertex $m_j$ is adjacent with exactly two vertices in the variable part (one in $\bigcup_i L_i^{\uparrow}$, one in $\bigcup_i L_i^{\downarrow}$), which results in $2\ceil{n^{1/3}}$  incident edges in total.
Moreover, for every variable $x\in \middle^{-1}(m_j)$ vertex $m_j$ has at most $4$ edges to the clause part, since $x$ occurs in at most $4$ clauses.
It follows that $\deg_G(m_j) \le 6\ceil{n^{1/3}}$.
Now we focus on the clause part.
Since every cluster has $O(n^{1/3})$ vertices, and there are no edges between the clusters, every vertex in the graph induced by the clause part has degree $O(n^{1/3})$.
Vertices in  $\bigcup_i A_i \cup  \bigcup_i C_i$ have no more edges in $G$.
It remains to consider an arbitrary vertex $b_{i,j}^k$ and count the edges connecting it to the variable part.
These edges are of the form $b_{i,j}^k\middle(x)$, for some variable $x$ that occurs in a clause $C$ such that $\beta(C)=i$ and $\gamma_i(C)=j$.
Since $|\gamma_i^{-1}(j)|\le n_i$, there are $O(n^{1/3})$ such clauses $C$, and each of them contains three variables.
Hence there are $O(n^{1/3})$ edges connecting $b_{i,j}^k$ and the variable part and $\deg_G(b_{i,j}^k) = O(n^{1/3})$.
To sum up, $\Delta(G)=O(n^{1/3})$, as required.

\heading{Maximum degree of $G_S=(V,S)$}
Let us inspect each kind of vertices in $V$ separately.

First consider a vertex $u_{i,j}\in L_i^{\uparrow}$ in the variable part.
It is incident with two kinds of edges in $G_S$.
The edges of the first kind are of the form $\{u_{i,j}, l_{i,\low(x)}\}$, for some variable $x$ such that $\lay(x)=i$ and $\up(x)=j$.
Since $|\up^{-1}(j)| \le \ceil{n^{1/3}}+3$, we get that $u_{i,j}$ is incident with at most $\ceil{n^{1/3}}+3$ edges of the first kind.
The edges of the second kind are of the form $\{b^k_{\beta(C),\gamma(C)},u_{i,j}\}$, for some variable $x$ in $k$-th literal of a clause $C$ such that $\lay(x)=i$ and $\up(x)=j$.
Since $|\up^{-1}(j)| \le \ceil{n^{1/3}}+3$ and $x$ is in at most 4 clauses, we get that $u_{i,j}$ is incident with at most $4(\ceil{n^{1/3}}+3)$ edges of the second kind.
It follows that $\deg_{G_S}(u_{i,j}) = O(n^{1/3})$.
In an analogous way we can bound the degree of vertices in sets $L_i^{\downarrow}$.

Consider a vertex $m_j \in M$ in the variable part.
It is incident with edges of the form $\{m_j, a_{\beta(C),g(C)}\}$, for some variable $x$ from a clause $C$ such that $\middle(x)=m_j$.
Since $|\middle^{-1}(m_j)| \le \ceil{n^{1/3}}$ and every variable is in at most 4 clauses, we get $\deg_{G_S}(m_j) = O(n^{1/3})$.

Consider a vertex $a_{i,j}$ in the clause part.
It is incident with at most $n_i=O(n^{1/3})$ edges of the form $\{a_{i,j},c_{i,\gamma(C)}\}$, since $|C_i|=n_i$.
It is also incident with edges of the form $\{\middle(x),a_{i,j}\}$, where $x$ is a variable in a clause $C$ such that $\beta(C)=i$ and $g(C)=j$.
Since $g$ is injective on every color class of $\gamma_i$, so $|g^{-1}(j) \cap \beta^{-1}(i)|$ is bounded by the number of colors in $\gamma_i$, which is $O(n^{1/3})$.
Hence there are $O(n^{1/3})$ clauses $C$ with $g(C)=j$ and each of them contains three variables, so $\deg_{G_S}(a_{i,j}) = O(n^{1/3})$.

Consider a vertex $b^k_{i,j}$ in the clause part.
It is incident with at most $|\gamma^{-1}(j)|\le n^{1/3}$ edges of the form $\{b^k_{i,j},u_{\lay(x),\up(x)}\}$ or $\{b^k_{i,j},l_{\lay(x),\low(x)}\}$, where $x$ is the variable in the $k$-the literal of a clause $C$ such that $\beta(C)=i$ and $\gamma(C)=j$. 
It follows that  $\deg_{G_S}(b^k_{i,j}) \le n^{1/3}$.

Finally consider a vertex $c_{i,j}$ in the clause part.
It is incident with at most $\ceil{n^{1/3}}$ edges of the form $\{a_{i,g(C)},c_{i,j}\}$, since $|A_i|=\ceil{n^{1/3}}$.
Hence $\deg_{G_S}(c_{i,j}) \le \ceil{n^{1/3}}$.

To sum up, $\Delta(G_S)=O(n^{1/3})$, as required.

\heading{Additional properties}
We have $|\Dom(c_0)| = \sum_i3|C_i| = O(n^{2/3})$.
Notice that the vertices of $(V(G), E \cup S)$ can be partitioned into four independent sets as follows:
$I_1 = \bigcup_i L_i^{\uparrow} \cup \bigcup_i A_i$,
$I_2 = M$,
$I_3 = \bigcup_i B_i$,
$I_4 = \bigcup_i L_i^{\downarrow} \cup \bigcup_i C_i$.
This defines the desired $4$-coloring of $(V(G),E \cup S)$.
Now consider the precoloring conflict graph $CG_{G,S,c_0}$.
Notice that for every $i=1,\ldots,s$ cluster $Q_i$ has $3n_i = O(n^{1/3})$ colored edges.
For every $i=1,\ldots,s$, color the vertices of $CG_{G,S,c_0}$ corresponding to the colored edges of $Q_i$ using different colors from $1$ to $3n_i$.
Notice that two colored edges $e_1,e_2\in\Dom(c_0)$ that belong to different clusters are not adjacent in the precoloring conflict graph $CG_{G,S,c_0}$.
Hence we defined a proper $O(n^{1/3})$-coloring of $CG_{G,S,c_0}$.

\heading{From an assignment to a coloring}
Let $\xi:\Var\rightarrow\{T,F\}$ be a satisfying assignment of $\varphi$.
We claim that there is a coloring $c$ of $E(G)$ which extends $c_0$ and satisfies all pairs in $S$.
We define $c$ as follows.
Denote $\overline{F}=T$, $\overline{T}=F$ and $\xi(\overline{x}) = \overline{\xi(x)}$.
For every variable $x\in\Var$ we put $c(u_{\lay(x),\up(x)}\middle(x))=\xi(x)$ and $c(\middle(x)l_{\lay(x),\low(x)})=\overline{\xi(x)}$.
By $(P_3)$ and $(P_4)$ each edge is colored exactly once.
Note that it satisfies all the pairs in $S$ between vertices in the variable part.

For each clause $C$ and each of its literals $\ell_k$ do the following.
Color the edge $a_{\beta(C),g(C)}b_{\beta(C),\gamma(C)}^k$ with the color $\xi(\ell_k)$.
Since $g$ is injective on color classes of $\gamma_{\beta(C)}$, after processing all the literals in all the clauses, no edge is colored more than once. 
Recall that for every clause $C$ we added exactly one pair to $S$, namely $\{a_{\beta(C),g(C)},c_{\beta(C),\gamma(C)}\}$.
Pick any of $C$'s satisfied literals, say $\ell_k$. 
Note that the pair $\{a_{\beta(C),g(C)},c_{\beta(C),\gamma(C)}\}$ is then satisfied, because edge $a_{\beta(C),g(C)}b_{\beta(C),\gamma(C)}^k$ is colored by $T$ and $b_{\beta(C),\gamma(C)}^kc_{\beta(C),\gamma(C)}$ is colored by $F$.
Hence all the pairs in $S$ between vertices in the clause part are satisfied.

Now let us color the edges between the clause part and the variable part.
Consider any such edge $uv$, i.e., $u$ is in the clause part and $v$ is in the variable part.
By $(P_6)$, there is exactly one clause $C$ and exactly one literal $\ell_k \in C$ such that $u=b_{\beta(C),\gamma(C)}^k$ and $v=\middle(x)$, where $x$ is the variable in $\ell_k$.
Color the edge $b_{\beta(C),\gamma(C)}^k \middle(x)$ with the color $\overline{\xi(\ell_k)}$.
Then the pair $\{\middle(x), a_{\beta(C),g(C)}\}$ is satisfied by the path $(\middle(x), b_{\beta(C),\gamma(C)}^k, a_{\beta(C),g(C)})$, since $c(b_{\beta(C),\gamma(C)}^k a_{\beta(C),g(C)})=\xi(\ell_k)$.
Assume $\ell_k=x$. Then the pair $\{b_{\beta(C),\gamma(C)}^k, u_{\lay(x),\up(x)}\}$ is satisfied by the path $(b_{\beta(C),\gamma(C)}^k, \middle(x), u_{\lay(x),\up(x)})$, since its first edge is colored by 
$\overline{\xi(\ell_k)}=\overline{\xi(x)}$ and its second edge is colored by $\xi(x)$.
Analogously, when $\ell_k=\bar{x}$, then the pair $\{b_{\beta(C),\gamma(C)}^k, l_{\lay(x),\low(x)}\}$ is satisfied by the path $(b_{\beta(C),\gamma(C)}^k, \middle(x), l_{\lay(x),\low(x)})$, since its first edge is colored by $\overline{\xi(\ell_k)}=\xi(x)$ and its second edge is colored by $\overline{\xi(x)}$.

It follows that we colored all the edges and all the pairs in $S$ are satisfied, so $(G,S,c_0)$ is a YES-instance, as required.

\heading{From a coloring to an assignment}
Let $c:E(G)\rightarrow \{T,F\}$ be a coloring which extends $c_0$ and satisfies all pairs in $S$.
Consider the following variable assignment: for every $x\in \Var$, we put $\xi(x) = c(u_{\lay(x),\up(x)}\middle(x))$.
We claim that $\xi$ satisfies all the clauses of $\varphi$.
Consider an arbitrary clause $C=\{\ell_1,\ell_2,\ell_3\}$.

Since the pair $\{a_{\beta(C),g(C)},c_{\beta(C),\gamma(C)}\}$ is satisfied, there is a 2-color 2-path $P$ between $a_{\beta(C),g(C)}$ and $c_{\beta(C),\gamma(C)}$.
Recall that $N(c_{\beta(C),\gamma(C)}) = \{b_{\beta(C),\gamma(C)}^k\ :\ k=1,2,3\}$, so there is $k=1,2,3$ such that $b_{\beta(C),\gamma(C)}^k$ is the internal vertex on $P$.
Since $c$ extends $c_0$ and $c_0(b_{\beta(C),\gamma(C)}^k c_{\beta(C),\gamma(C)}) = F$, we infer that $c(a_{\beta(C),g(C)} b_{\beta(C),\gamma(C)}^k) = T$.
Let $x$ be the variable in the literal $\ell_k$. 

Since the pair $\{\middle(x), a_{\beta(C),g(C)}\}$ is satisfied, there is a 2-color 2-path $Q$ between $\middle(x)$ and $a_{\beta(C),g(C)}$.
Then the internal vertex of $Q$ is $b_{\beta(C'),\gamma(C')}^{k'}$, for some clause $C'$ and integer $k'=1,2,3$.
Let $y$ be the variable in the $k'$-th literal of $C'$. 
Since there is an edge between $\middle(x)$ and $b_{\beta(C'),\gamma(C')}^{k'}$, from $(P_6)$ we infer that $\middle(y)=\middle(x)$.
If $C=C'$ and $k'\ne k$, then by $(P_1)$ we get that $\middle(x)\ne\middle(y)$, a contradiction.
If $C\ne C'$, since $\middle(y)=\middle(x)$, the clauses $C$ and $C'$ are adjacent in the clause conflict graph $G_C$, so $\beta(C')\ne \beta(C)$.
However, then the edge $b_{\beta(C'),\gamma(C')}^{k'} a_{\beta(C),g(C)}$ of $Q$ goes between two clusters, a contradiction.
Hence $C'=C$ and $k'=k$, i.e., $Q=(\middle(x), b_{\beta(C),\gamma(C)}^k, a_{\beta(C),g(C)})$. Since $c(b_{\beta(C),\gamma(C)}^k a_{\beta(C),g(C)}) = T$, we get $c(\middle(x) b_{\beta(C),\gamma(C)}^k) = F$.
Now assume w.l.o.g. that $\ell_k=x$, the case $\ell_k=\bar{x}$ is analogous.

Since the pair $\{b_{\beta(C),\gamma(C)}^k, u_{\lay(x),\up(x)}\}$ is satisfied, there is a 2-color 2-path $R$ between $b_{\beta(C),\gamma(C)}^k$ and $u_{\lay(x),\up(x)}$.
Then the internal vertex $z$ of $R$ belongs to $M$.
By $(P_6)$ there is a literal $\ell_k$ which belongs to a clause $C_2$ and contains a variable $x_2$ such that $z=\middle(x_2)$ and $b_{\beta(C),\gamma(C)}^k = b_{\beta(C_2),\gamma(C_2)}^k$.
In particular, $\beta(C)=\beta(C_2)$ and $\gamma(C)=\gamma(C_2)$.
Assume $C_2 \ne C$.
There is a variable, say $x_3$, corresponding to edge $\middle(x_2)u_{\lay(x),\up(x)}$, i.e., $\middle(x_2)=\middle(x_3)$ and $u_{\lay(x),\up(x)} = u_{\lay(x_3),\up(x_3)}$.
It follows that $C$ and $C_2$ are adjacent in $G_{\beta(C)}$, which contradicts the fact that $\gamma(C)=\gamma(C_2)$.
Hence $C_2=C$, i.e., there is exactly one 2-path between $b_{\beta(C),\gamma(C)}^k$ and $u_{\lay(x),\up(x)}$, and it goes through $\middle(x)$.
Since  $c(\middle(x) b_{\beta(C),\gamma(C)}^k) = F$ and the path is 2-color, we get that $c(u_{\lay(x),\up(x)}\middle(x))=T$.
Hence $\xi(\ell_k) = \xi(x) = T$, so clause $C$ is satisfied, as required.
It finishes the proof.
\end{proof}

\subsection{From $2$ colors to $k$ colors}
\label{sec:PartRTwoCExt->PartRCExt}

\begin{lemma} \label{lem:from-PartRTwoCExt-to-PartRCExt}
	For any fixed $k \ge 3$, there is a polynomial time algorithm which given an instance $I=(G=(V,E), S, c_0)$
	of \probPartRTwoCExt constructs	an equivalent instance
	$I'=(G'=(V', E'), S', c_0')$ of \probPartRCExt such that
	$|V'| = O(k |V|)$,
	$|E'| = |E| + O(k |V|)$,
	$|S'| = |S| + |E|$,
	$\Delta(G') \le \Delta_1(G) + 2$
	and
	$|\Dom(c_0')| = |\Dom(c_0)| + O(k|V|)$.
	Let $G_S = (V, S)$ and $G_{S'} = (V', S')$.
	Then $\Delta(G_{S'}) \le \Delta(G) + \Delta(G_S)$.
	Moreover, given a proper vertex $p$-coloring of $G_S$ the algorithm outputs also a $(p+1)$-coloring of $G_{S'}$.
	Also, given a proper vertex $q$-coloring of $G_{E\cup S} = (V, E \cup S)$ and a proper vertex $\ell$-coloring of the precoloring conflict graph $CG_I$ the algorithm
	outputs a proper $(\ell + O(q))$-coloring of the precoloring conflict graph $CG_{I'}$.
\end{lemma}

\begin{proof}
\heading{Construction}
	Let us denote the colors of $c_0$ by $1$ and $2$.
	Let
	\[V' = V \cup \bigcup_{v\in V} \{v_1^1, v_1^2, v_2, v_3, \ldots, v_{k - 1}\},\]
	\[E' = E \cup \bigcup_{v \in V} \{
	v v_1^1, v_1^1 v_2, v v_1^2, v_1^2 v_2, v_2 v_3, v_3 v_4, \ldots,
	v_{k - 2} v_{k - 1} \},\]
	\[S' = S \cup \{u_{k - 1} v \ :\  uv \in E \text{ and } u < v\}\]
	and $\Dom(c_0') = \Dom(c_0) \cup (E' \setminus E)$.
	Let $c_0'|_{\Dom(c_0)} = c_0$
	and let for every $v \in V$
	$c_0'(v v_1^1) = 1$,
	$c_0'(v v_1^2) = 2$,
	$c_0'(v_1^1 v_2) = 3$,
	$c_0'(v_1^2 v_2) = 3$,
	and $c_0'(v_i v_{i + 1}) = i + 2$
	for every $i \in \{2, 3, \ldots, k - 2\}$.
	(See Fig~\ref{fig:antena}.)
	Note that for every vertex $u_{k - 1}$ such that $u \in V$ we have $\deg_{G_{S'}}(u_{k - 1})\le\deg_G(u)$ and for every vertex $v \in V$ we have $\deg_{G_{S'}}(v)\le\deg_{G_S}(v)+\deg_G(v)$.
	Hence $\Delta(G_{S'}) \le \Delta(G) + \Delta(G_S)$, as required.
	
	\begin{figure}
	\begin{center}
	\begin{tikzpicture}
\tikzset{vtx/.style={draw, circle, fill=black, line width=1pt, inner sep = 2pt}}
\tikzset{edgecolor/.style={draw, rectangle, line width=.5pt, fill=white, pos=0.7, inner sep=1pt}}

\node [vtx,label={[label distance=1pt]90:$v$}] (v) at (0,0) {};
\node [vtx,label={[label distance=1pt]90:$v^1_1$}] (v11) at (1,1) {};
\node [vtx,label={[label distance=1pt]90:$v^2_1$}] (v12) at (1,-1) {};
\node [vtx,label={[label distance=1pt]90:$v_2$}] (v2) at (2,0) {};
\node [vtx,label={[label distance=1pt]90:$v_3$}] (v3) at (4,0) {};
\node [vtx,label={[label distance=1pt]90:$v_4$}] (v4) at (6,0) {};
\node at (7,0) {$\ldots$};
\node [vtx,label={[label distance=1pt]90:$v_{k-2}$}] (vk2) at (8,0) {};
\node [vtx,label={[label distance=1pt]90:$v_{k-1}$}] (vk1) at (10,0) {};

\draw [thick] (v) edge node [edgecolor,pos=0.5] {1} (v11);
\draw [thick] (v) edge node [edgecolor,pos=0.5] {2} (v12);
\draw [thick] (v2) edge node [edgecolor,pos=0.5] {3} (v11);
\draw [thick] (v2) edge node [edgecolor,pos=0.5] {3} (v12);
\draw [thick] (v2) edge node [edgecolor,pos=0.5] {4} (v3);
\draw [thick] (v3) edge node [edgecolor,pos=0.5] {5} (v4);
\draw [thick] (vk1) edge node [edgecolor,pos=0.5] {$k$} (vk2);

\end{tikzpicture}
	\end{center}
\caption{\label{fig:antena}The gadget added to every vertex $v\in V$ in Lemma~\ref{lem:from-PartRTwoCExt-to-PartRCExt}.}
	\end{figure}
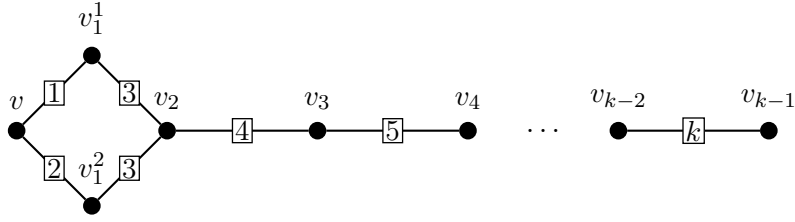

	\heading{Equivalence}
	Assume $(G=(V, E), S, c_0)$ is a YES-instance of \probPartRTwoCExt and let $c$ be the corresponding $2$-coloring.
	Define a coloring $c'$ of $E(G')$ as $c'|_E = c$ and	$c'|_{E' \setminus E} = c_0'|_{E' \setminus E}$.
	Note that $c'$ is an extension of $c_0'$.
	Moreover, $c'$ satisfies all the pairs of $S$ using rainbow paths in $E$ because $c$ satisfies $S$.
	All the pairs of the form $\{u_{k - 1}, v\}$ for some $uv \in E$ where $u < v$ are satisfied because one of the paths
	$v u u_1^1 u_2 u_3 \ldots u_{k - 1}$ and $v u u_1^2 u_2 u_3 \ldots u_{k - 1}$ is rainbow.
	Hence $(G',S',c_0')$ is a YES-instance of \probPartRCExt.
	
	Now assume that	$(G'=(V', E'), S', c_0')$ is a YES instance of \probPartRCExt and let $c'$ be the corresponding $k$-coloring.
	Then we claim that $c = c'|_E$ is a coloring of $E(G)$ that satisfies $S$ and extends $c_0$.
	It is easy to see that $c$ extends $c_0$, because $c_0'$ extends $c_0$ and $c'$ extends $c_0'$.

	Now we show that $c(E) \subseteq \{1, 2\}$. 
	Indeed, for every edge $uv\in E$ such that
	$u < v$ the distance between the vertices $u_{k - 1}$ and $v$ is $k$.
	Hence all rainbow paths between $u_{k - 1}$ and $v$ are of length exactly $k$.
	There are exactly two paths of length $k$ between them, namely $v u u_1^1 u_2 u_3 \ldots u_{k - 1}$ and $v u u_1^2 u_2 u_3 \ldots u_{k - 1}$.
	Exactly one of them is rainbow,	so $c(uv) = 2$ (if the first one is rainbow) or $c(uv) = 2$ (otherwise), as required.
	
	Finally we show that all pairs in $S$ are satisfied by $c$.
	Pick any $\{u,v\}\in S$.
	Note that no $(u,v)$-path in $G'$ visits a vertex from $V'\setminus V$ because every vertex $v\in V$ separates the vertices in $\{v_1^1, v_1^2, v_2, v_3, \ldots, v_{k - 1}\}$ from the rest of the graph.
	It follows that the rainbow path that satisfies $\{u,v\}$ in $c'$ also satisfies $\{u,v\}$ in $c$.
	This finishes the proof that $(G, S, c_0)$ is a YES-instance of \probPartRTwoCExt.

	\heading{Additional Properties}
	The vertex $(p + 1)$-coloring of $G_{S'} = (V', S')$ can be easily obtained from the
	$p$ coloring of $G_S = (V, S)$.	Indeed, all the independent sets of
	$G_S = (V, S)$ are also independent in $G_{S'} = (V', S')$
	and all the added vertices $V'\setminus V$ form
	an independent set in $G_{S'} = (V', S')$.
	Therefore it is sufficient to extend the input $p$-coloring with one additional color for the vertices
	of $V' \setminus V$.
	
	Let $\alpha:\Dom(c_0)\rightarrow [\ell]$ be the given vertex $\ell$-coloring of $CG_I$, and let $h:V\rightarrow [q]$ be the given vertex $q$-coloring of the graph $G_{E \cup S} = (V, E\cup S)$.
	Note that in $G'$ we have not added any edge or requirement between any two vertices of $V$.
	Therefore $CG_{I'}[\Dom(c_0)] = CG_I$.
	Hence it suffices to extend $\alpha$ to a coloring $\alpha'$ of $CG_{I'}$.
	The remaining elements, i.e., elements from $\Dom(c'_0)\setminus \Dom(c_0)$ get new colors, from the set $\{\ell + 1, \ell + 2, \ldots \ell + O(q)\}$, as follows.
	For every vertex $v\in V$ we put $\alpha'(v v_1^1) =\ell + 2 h(v) - 1$ and $\alpha'(v v_1^2)=\ell + 2 h(v)$.
	Thus we added $2q$ new colors.
	Note that their corresponding color classes are independent sets in $CG_{I'}$.
	In what follows they stay independent because will not use these $2q$ colors any more.
	Next, for every vertex $v \in V$ we put $\alpha'(v_1^1 v_2)=\ell + 2q + 1$ and $\alpha'(v_1^2v_2)=\ell + 2q + 2$.
	Again, the corresponding color classes are independent in $CG_{I'}$.
	If $k = 3$ then we have just properly colored all the vertices of $CG_{I'}$, i.e., the edges of $\Dom(c'_0)$.
	If $k > 3$ then for every vertex $v \in V$ we color the remaining path $v_3 v_4 \ldots v_{k - 1}$ using
	alternating sequence of colors
	$\ell + 2q + 3, \ell + 2q + 4, \ell + 2q + 5,
	\ell + 2q + 3, \ell + 2q + 4, \ldots$.
	Thus we have obtained a proper vertex $(\ell + 2q + 5)$-coloring of the graph $CG_{I'}$, as required.
\end{proof}

\subsection{From \probPartRCExt to \probPartRC}
\label{sec:PartRCExt->PartRC}

\begin{lemma} \label{lem:from-PartRCExt-to-PartRC}
	For any fixed $k \ge 3$, given an instance $(G=(V, E), S, c_0)$
	of \probPartRCExt
	together with a proper vertex
	$\ell$-coloring of the precoloring
	conflict graph $CG_I$
	one can construct in polynomial time an equivalent
	instance $(G'=(V', E'), S')$ of \probPartRC such that
	$|V'| = |V| + O(k^2\ell)$,
	$|E'| = |E| + |\Dom(c_0)| + O(k^2\ell)$,
	$|S'| = |S| + 2|\Dom(c_0)| + O(k^2\ell)$.
	Let $G_S = (V, S)$ and $G_{S'} = (V', S')$.
	Then $\Delta(G_{S'}) = O(\Delta(G_S) + \Delta(G) + |\Dom(c_0)| / \ell)$.
	Moreover if we are given a proper vertex $p$-coloring of the graph $G_S = (V, S)$ then we can output also a proper vertex $(p + 2)$-coloring
	of the graph $G_{S'} = (V', S')$.
\end{lemma}

\begin{proof}
        \heading{Construction}
	Let $f:\Dom(c_0) \to [\ell]$ be the vertex $\ell$-coloring of the precoloring conflict graph $CG_I$.
	If there is a color class larger than $\lceil |\Dom(c_0)| / \ell \rceil$ we split it into two colors:
	one of size $\lceil \Dom(c_0) / \ell \rceil$ and the rest.
	We repeat this procedure until none of the color classes is larger than $\lceil |\Dom(c_0)| / \ell \rceil$.
	This process introduces at most $\ell$ new colors, since $\ell \cdot \lceil |\Dom(c_0)| / \ell \rceil \ge |\Dom(c_0)|$.
	Hence in what follows we assume that $f$ is an $\ell'$-coloring of $CG_I$ such that each color class is of size at most $\lceil |\Dom(c_0)| / \ell \rceil$, where $\ell'\le 2\ell$.

	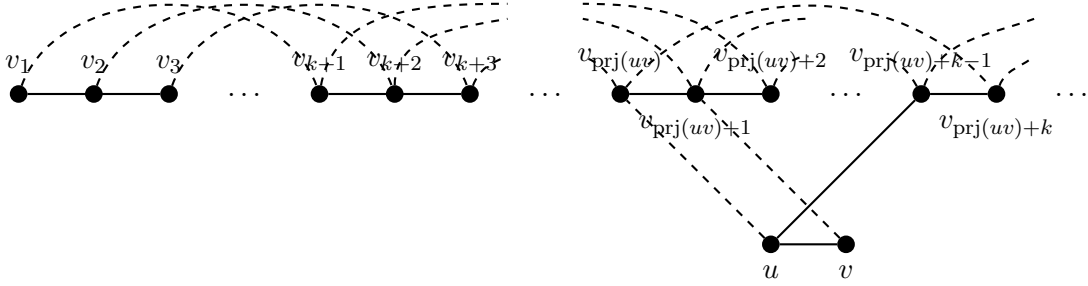
\begin{figure}
	\begin{center}
	 \begin{tikzpicture}
\tikzset{vtx/.style={draw, circle, fill=black, line width=1pt, inner sep = 2pt}}
\tikzset{edgecolor/.style={draw, rectangle, line width=.5pt, fill=white, pos=0.7, inner sep=1pt}}

\node [vtx,label={[label distance=1pt]90:$v_1$}] (v1) at (0,0) {};
\node [vtx,label={[label distance=1pt]90:$v_2$}] (v2) at (1,0) {};
\node [vtx,label={[label distance=1pt]90:$v_3$}] (v3) at (2,0) {};
\node at (3,0) {$\ldots$};
\node [vtx,label={[label distance=1pt]90:$v_{k+1}$}] (vk1) at (4,0) {};
\node [vtx,label={[label distance=1pt]90:$v_{k+2}$}] (vk2) at (5,0) {};
\node [vtx,label={[label distance=1pt]90:$v_{k+3}$}] (vk3) at (6,0) {};
\node at (7,0) {$\ldots$};
\node [vtx,label={[label distance=1pt]90:$v_{{\rm prj}(uv)}$}] (vp) at (8,0) {};
\node [vtx,label={[label distance=1pt]-90:$v_{{\rm prj}(uv)+1}$}] (vp1) at (9,0) {};
\node [vtx,label={[label distance=1pt]90:$v_{{\rm prj}(uv)+2}$}] (vp2) at (10,0){};
\node at (11,0) {$\ldots$};
\node [vtx,label={[label distance=1pt]90:$v_{{\rm prj}(uv)+k-1}$}] (vpp) at (12,0) {};
\node [vtx,label={[label distance=1pt]-90:$v_{{\rm prj}(uv)+k}$}] (vpp1) at (13,0) {};
\node at (14,0) {$\ldots$};
\node [vtx,label={[label distance=1pt]-90:$u$}] (u) at (10,-2) {};
\node [vtx,label={[label distance=1pt]-90:$v$}] (v) at (11,-2) {};

\draw [thick] (v1) -- (v3);
\draw [thick] (vk1) -- (vk3);
\draw [thick] (vpp) -- (vpp1);
\draw [thick] (vp) -- (vp2);
\draw [thick] (v) -- (u);
\draw [thick] (u) -- (vpp);
\draw [thick,dashed] (u) -- (vp);
\draw [thick,dashed] (v) -- (vp1);
\draw [thick,dashed] (v1) to[out=70,in=110] (vk1);
\draw [thick,dashed] (v2) to[out=70,in=110] (vk2);
\draw [thick,dashed] (v3) to[out=70,in=110] (vk3);
\draw [thick,dashed] (vp) to[out=50,in=130] (vpp1);
\draw [thick,dashed] (vk1) to[out=80,in=180] (6.5,1.2);
\draw [thick,dashed] (vk2) to[out=90,in=190] (6.5,1);
\draw [thick,dashed] (vk3) to[out=70,in=200] (6.5,0.5);
\draw [thick,dashed] (vp) to[out=120,in=-20] (7.5,0.5);
\draw [thick,dashed] (vp1) to[out=120,in=-10] (7.5,1);
\draw [thick,dashed] (vp2) to[out=120,in=0] (7.5,1.2);
\draw [thick,dashed] (vp1) to[out=70,in=180] (10.5,1);
\draw [thick,dashed] (vp2) to[out=70,in=180] (10.5,0.5);
\draw [thick,dashed] (vpp) to[out=120,in=-20] (11.5,0.5);
\draw [thick,dashed] (vpp) to[out=70,in=200] (13.5,1);
\draw [thick,dashed] (vpp1) to[out=70,in=220] (13.5,0.5);

\end{tikzpicture}
\caption{\label{fig:path}Construction in Lemma~\ref{lem:from-PartRCExt-to-PartRC}. Dashed lines denote requests.}
	\end{center}
	\end{figure}

	Let $V' = V \cup \{v_i :\  i \in [3k^2 \ell']\}$.
        For every $\alpha \in [\ell']$, $\beta \in [k]$ and $\gamma \in [3k]$ let us denote $\id(\alpha,\beta,\gamma) = (\alpha-1)\cdot 3k^2 + (\beta - 1)\cdot 3k + \gamma$.
	Note that for every $i\in [3k^2\ell']$ there is exactly one triple
	$(\alpha, \beta, \gamma) \in [\ell'] \times [k] \times [3k]$
	such that $i = \id(\alpha,\beta,\gamma)$.
        Moreover, for every edge $uv\in E$, such that $u < v$, we denote $\prj(uv) = \id(f(uv),c_0(uv),c_0(uv))$.
	Let
	\[
	E' = E
		\cup \{v_i v_{i + 1}
			:\  i \in [3k^2\ell' - 1] \}
		\cup \{v_{\prj(uv) + k - 1} u
			:\  uv\in \Dom(c_0) \text{ and } u < v\}\]
	and
	\[S' = S
		\cup \{
			\{v_i, v_{i + k}\}
			:\  i \in [3k^2\ell' - k]
		\}
		\cup \{
			\{v_{\prj(uv)}, u\},
			\{v_{\prj(uv) + 1}, v\}
			:\  uv\in \Dom(c_0) \text{ and } u < v\}.\]	
	From the above, if we have $uv, u'v' \in \Dom(c_0)$
	such that
	$u < v$, $u' < v'$
	and $\prj(uv) = \prj(u'v')$
	then $f(uv) = f(u'v')$
	and $c_0(uv) = c_0(u'v')$.	
			
        \heading{Equivalence}
	Assume $c$ is a $k$-coloring of $E$ that satisfies all the
	constraints of $S$ and extends the coloring $c_0$.
	We define a coloring $c':E' \to [k]$ as follows.
	For every edge $xy \in E$ we put $c'(xy)=c(xy)$.
	Let us define
	\newcommand{\bmodone}{{\ \mbox{mod}_1\ }}
	\[x \bmodone y = 1 + (x - 1) \bmod y.\]
	The edges of the path $(v_1,\ldots,v_{3k^2 \ell'})$ are colored with the sequence $(1,\ldots,k)$ repeated, i.e., for every $i=1,\ldots,3k^2 \ell'-1$ we put $c'(v_iv_{i+1}) = i \bmodone k$.
	Finally, for every  edge $uv\in E$, such that $u < v$, we put 
	\[c'(v_{\prj(uv) + k - 1} u) = (c_0(uv) - 1)\bmodone k.\]
	We claim that $c'$ is a $k$-coloring of $E'$ that satisfies all the constraints	of $S'$.
	Indeed, every constraint of $S$ is satisfied because $c'$ extends $c$.
	Every constraint of the form $\{v_i, v_{i + k}\}$ is satisfied because of the path $(v_i, v_{i + 1}, \ldots v_{i + k})$.
	Finally, note that
	\[c'(v_{\prj(uv)}, v_{\prj(uv) + 1}) =
		c_0(uv) \bmodone k = c_0(uv)\]
	and the path
	$(v_{\prj(uv)}, v_{\prj(uv)+1}, \ldots v_{\prj(uv) + k - 1})$
	uses all the colors except for 
	$(\prj(uv) + k - 1) \bmodone k = (c_0(uv) - 1) \bmodone k$.
	But $c'(v_{\prj(uv) + k - 1} u) = (c_0(uv) - 1) \bmodone k$.
	So, for every $uv \in \Dom(c_0)$ such that $u < v$ the constraints $\{v_{\prj(uv)}, u\}$ and $\{v_{\prj(uv) + 1}, v\}$
	are satisfied because of the paths $P_{uv}^1=(v_{\prj(uv)}, v_{\prj(uv)+1}, \ldots v_{\prj(uv) + k - 1}, u)$ and $P_{uv}^2=(v_{\prj(uv) + 1}, v_{\prj(uv) + 2}, \ldots v_{\prj(uv) + k - 1}, u, v)$, respectively.
	
	Now assume $c'$ is a $k$-coloring of $E'$ that satisfies all the constraints of $S'$.
	The distance between $i = \id(a, s, s)$ and $j = \id(a, s + 1, s + 1)$
	is equal $3k + 1$.
	Moreover the distance between $i = \id(a, k, k)$ and $j = \id(a + 1, 1, 1)$
	is equal $2k + 1$.
	So if for some different $i,j=1,\ldots,3k^2\ell'$
	such that $i < j$
	both $i$ and $j$ have neighbors in $V$, then $|i-j|\ge 2k+1$.
	This has three consequences:
	
	\begin{enumerate}[$(i)$]
	 \item for any $i \in [3k^2\ell' - k]$, there is exactly one path of length at most $k$ between $v_i$ and $v_{i + k}$, namely $(v_i,v_{i+1},\ldots,v_{i + k})$,
	 \item for any $uv\in \Dom(c_0)$, $u < v$, there is exactly one path of length at most $k$ between $v_{\prj(uv)}$ and $u$, namely $P_{uv}^1$,
	 and exactly one path of length at most $k$ between $v_{\prj(uv)+1}$ and $v$, namely $P_{uv}^2$,
	 \item for any pair of vertices $u,v\in V$ every path of length at most $k$ between $u$ and $v$ does not contain any edge $v_iv_{i+1}$, for $i\in [3k^2\ell' - 1]$.
	\end{enumerate}

	From $(i)$ it follows that $(c'(v_1v_2),c'(v_2v_3),\ldots,c'(v_kv_{k+1}))$ is a permutation of all $k$ colors, and this permutation repeats $3k\ell'$ times, i.e., for every $i \in [3k^2\ell' - k - 1]$ we have $c'(v_i v_{i + 1}) = c'(v_{k + i} v_{k + i + 1})$.
	Let $c''$ be the coloring of $E'$ obtained from $c'$ by permuting the colors so that for every $i=1,\ldots,3k^2 \ell'-1$ we have $c''(v_iv_{i+1}) = i \bmodone k$.
	Obviously, $c''$ satisfies $S'$, since $c'$ does.
	We claim that the coloring $c = c''\mid_E$ satisfies all the constraints of $S$ and extends $c_0$.
	
	Consider an arbitrary edge $uv \in \Dom(c_0)$, with $u < v$.
	Note that $c''(v_{\prj(uv)}, v_{\prj(uv)+1}) = \prj(uv) \bmodone k = c_0(uv)$.
	Hence the path $(v_{\prj(uv)}, v_{\prj(uv)+1}, \ldots v_{\prj(uv) + k - 1})$ uses all the colors except for
	$(c_0(uv)-1) \bmodone k$.
	From $(ii)$ and the fact that $\{v_{\prj(uv)}, u\} \in S'$ we infer that $P_{uv}^1$ is rainbow, which implies
	$c''(v_{\prj(uv) + k - 1}u)=(c_0(uv)-1) \bmodone k$.
	Hence the path $(v_{\prj(uv)+1}, \ldots v_{\prj(uv) + k - 1}, u)$ uses all the colors except for $c_0(uv)$.
	From $(ii)$ and the fact that $\{v_{\prj(uv)+1}, v\} \in S'$ we infer that $P_{uv}^2$ is rainbow, which implies that $c''(uv)=c_0(uv)$.
	This shows that $c''$ extends $c_0$, and hence so does $c$.
	
	In the previous paragraph we showed that for every edge $uv \in \Dom(c_0)$, with $u < v$, we have $c''(v_{\prj(uv) + k - 1}u)=(c_0(uv)-1) \bmodone k$.
	Note that if there is an edge $v_{\prj(uv) + k - 1}u'$ for some $u'\in V$, it means that there is $u'v'\in\Dom(c_0)$ such that $u' < v'$ and $\prj(uv)=\prj(u'v')$.
	Hence $c_0(u'v')=c_0(uv)$, and by the previous paragraph, $c''(v_{\prj(uv) + k - 1}u') = (c_0(u'v')-1) \bmodone k =(c_0(uv)-1) \bmodone k = c''(v_{\prj(uv) + k - 1}u)$.
	In other words, for every $i=1,\ldots,3k^2\ell'$, all the edges between $v_i$ and $V$ have the same color in $c''$.
	This, combined with $(iii)$ means that for every pair of vertices $u,v\in V$ every rainbow path in $c''$ is contained in $E$.
	Hence, for any $\{u,v\}\in S$ the rainbow path between them in $c''$ is also a rainbow path in $c$, so $c$ satisfies $\{u,v\}$.
	This ends the proof of the equivalence.
	
        \heading{Additional properties}
	Note that for every vertex $v \in V' \setminus V$
	we have $\deg_{G_{S'}}(v)\le 2 + \lceil |\Dom(c_0)| / \ell' \rceil$
	and for every vertex $v \in V$ we have $\deg_{G_{S'}}(v)\le \deg_{G_S}(v) + \deg_{G}(v)$.
	Hence  $\Delta(G_{S'}) = O(\Delta(G_S) + \Delta(G) + |\Dom(c_0)| / \ell)$, as required.

	Note that in the graph $G_{S'}[V'\setminus V]$ all the
	connected components are paths.
	Hence the vertices of $V'\setminus V$ can be colored
	using two colors in $G_{S'}$.
	By merging this coloring with the given $p$-coloring of the vertices
	of graph $G_S$ we obtain a proper vertex $(p + 2)$-coloring of $G_{S'}$.
\end{proof}

\subsection{From \probPartRC to \probRC}
\label{sec:PartRC->RC}

The basic idea of our reduction from \probPartRC to \probRC is to modify the graph so that the pairs of vertices from $\barE \setminus S$ can be somehow trivially satisfied, without affecting the satisfiability of $S$.
To this end we use a notion of biclique covering number (called also bipartite dimension).
The {\em biclique covering number} $\bc(G)$ of a graph $G$ is the smallest number of biclique subgraphs of $G$ that cover all edges of $G$.
The following proposition is well-known.

\begin{proposition}[Folklore]
\label{prop:clique}
It holds that
 $\bc(K_n) = \ceil{\log n}$, and the corresponding cover can be constructed in polynomial time.
\end{proposition}

\begin{proof}
Assume the vertex set of $K_n$ is $\{0,\ldots,n-1\}$.
The $i$-th biclique contains edges between the vertices that have $0$ at the $i$-th bit and the vertices that have $1$ at the $i$-th bit.
\end{proof}

Let $G=(V_1,V_2,E)$ be a bipartite graph. 
Then $\bibarG$ denotes the bipartite complement of $G$, i.e, the bipartite graph $(V_1,V_2,\{v_1v_2\ :\ \text{$v_1\in V_1$, $v_2\in V_2$, and $v_1v_2\not\in E$}\})$.
We will use the following result of Jukna.
Recall that we denote $\Delta_1(G)=\max\{\Delta(G),1\}$.

\begin{theorem}[Jukna~\cite{Jukna-on-set-intersection-representations}]
\label{thm:jukna}
For any $n$-vertex bipartite graph $G$ we have $\bc(\bibarG) = O(\Delta_1(G)\log n)$.
\end{theorem}

Let us call the cover from Theorem~\ref{thm:jukna} the {\em Jukna cover}.
In our application we need to be able to {\em compute} the Jukna cover fast.

\begin{lemma}
\label{lem:jukna-algo}
The Jukna cover can be constructed in $(i)$ expected polynomial time, or $(ii)$ deterministic $2^nn^{O(1)}$ time.
\end{lemma}

\begin{proof}
 Denote $\Delta=\Delta(G)$. If $\Delta=0$ the claim follows from Proposition~\ref{prop:clique}, so in what follows assume $\Delta\ge 1$.
 Jukna~\cite{Jukna-on-set-intersection-representations} shows a simple worst-case linear time algorithm which samples a biclique in $G$.
 Then it is proved that after sampling $t$ bicliques, the probability that there is an edge not covered by one of the bicliques is at most $n^2e^{-t/(\Delta e)}$.
 It follows that the probability that more than $\Delta e (2\ln n + 1)$ samples are needed is at most $e^{-1}$.
 If after $\Delta e (2\ln n + 1)$ samples some edges is not covered, we discard all the bicliques found and repeat the whole algorithm from the scratch.
 The expected number of such restarts is $1 / (1 - e^{-1})=O(1)$.
 
 Now we proceed to the second part of the claim.
 Let $G=(V_1,V_2,E)$.
 For every subset $A\subseteq V_1$ we define the biclique $B_A=(A,B,E_A)$, where $B$ is the set of vertices of $V_2$ adjacent in $\bibarG$ to all vertices of $A$.
 Clearly, $B_A$ is a subgraph of $\bibarG$ and for every subset $A\subseteq V_1$ it can be found in time linear in the size of $\bibarG$.
 Our deterministic algorithm works as follows: as long as not all edges of $\bibarG$ are covered, it picks the biclique $B_A$ which maximizes the number of new covered edges of $\bibarG$.
 Since all the bicliques in the set $\{B_A\ :\ A\subseteq V_1\}$ can be listed in time $O(2^n E(\bibarG))$, the total running time is $t2^nn^{O(1)}$, where $t$ is the size of the returned cover.
 It suffices to show that $t=O(\Delta\log n)$.
 
 Jukna~\cite{Jukna-on-set-intersection-representations} shows that if set $A$ is chosen by picking every vertex of $V_1$ independently with probability $\frac{1}{\Delta}$, then for any edge $uv\in E(\bibarG)$,
 
 \begin{equation}
 \label{eq:prob}
   \prob{uv \in E_A} \ge \frac{1}{\Delta e}.  
 \end{equation}
 
 Consider any step of our algorithm and let $R\subseteq E(\bibarG)$ be the set of the edges of $\bibarG$ which are not covered yet.
 By~\eqref{eq:prob} and the linearity of expectation a set $A$ sampled as described above covers at least $|R|/(\Delta e)$ new edges in expectation.
 In particular, it implies that there exists a set $A\subseteq V_1$ that covers at least $|R|/(\Delta e)$ new edges.
 Let $\alpha = (1-\frac{1}{\Delta e})^{-1}$.
 By Taylor expansion of $\log(1-x)$, it follows that  $t = O(\log_{\alpha} |E(\bibarG)|) = O(\log n / \log \alpha) = O(\Delta \log n)$.
\end{proof}

\begin{lemma} \label{lem:bliclique_covering}
  Let $G$ be an $n$-vertex graph with a given proper vertex $p$-coloring.
  Then the edges of $\barG$ can be covered by $O(p^2 \Delta_1(G) \log n)$ bicliques from $\barG$ so that any edge of $G$ and any biclique have at most one common vertex.
  This cover can be constructed in $(i)$ expected polynomial time, or $(ii)$ deterministic $2^nn^{O(1)}$ time.
\end{lemma}

\begin{proof}
The edges of $\barG$ between the vertices of any color class form a clique, so by Proposition~\ref{prop:clique} we can cover its edges using $O(\log n)$ bicliques.
If an edge of $G$ has both endpoints in such a biclique, these endpoints have the same color, contradiction.
For two different colors $i$ and $j$ the edges of $G$ between their color classes form a bipartite graph of maximum degree at most $\Delta(G)$.
Hence by Lemma~\ref{lem:jukna-algo} we can cover the edges of its bipartite complement using $O(\Delta_1(G)\log n)$ bicliques.
If an edge $uv$ of $G$ has both endpoints in such a biclique, then either (i) these endpoints have the same color, contradiction, or (ii) these endpoints belong to two different parts of the biclique, so $uv$ is in the biclique and hence $uv\in E(\barG)$, a contradiction.
Summing over all color classes and pairs of color classes, we use $O(p^2 \Delta_1(G) \log n)$ bicliques, as required.
\end{proof}

Now we proceed to the actual reduction. 
Somewhat surprisingly, the $k=2$ requires a slightly different construction than the $k\ge 3$ case, so we partition the proof into two lemmas.

\begin{lemma} \label{lem:from-PartRTwoC-to-RTwoC}
	Given an instance $(G=(V,E), S)$ of \probPartRTwoC
	together with a proper
	$p$-coloring of the graph $G_S=(V, S)$,
	one can construct an equivalent instance $G'$ of \probRTwoC
	such that
	$|V(G')| = O(
		|V|
		+ p^2 \Delta_1(G_S) \log |V|
	)$,
	$|E(G')| = O(
		|E(G)|
		+ (|V| + p^2 \Delta_1(G_S) \log |V|) \cdot p^2 \Delta_1(G_S) \log |V|
	)$.
	The construction algorithm can run in $(i)$ expected polynomial time, or $(ii)$ deterministic $2^{|V|}|V|^{O(1)}$ time.
\end{lemma}

\begin{proof}
	Let us consider a biclique
	covering of the complement
	of the graph $G_S$ with
	$q=O(p^2\Delta_1(G_S)\log n)$ bicliques
	$(U_1, V_1; E_1), (U_2, V_2; E_2), \ldots,
		(U_q, V_q; E_q)$
	as in Lemma~\ref{lem:bliclique_covering}.
	Let
	$W = \{w_1, w_2, \ldots, w_q\}$,
	$T = \{t_1, t_2, t_3\}$,
	$V(G') = V \cup W \cup T$
	and
	$E(G') =
		E(G)
		\cup (W \times W)
		\cup (T \times T)
		\cup (\{t_2\} \times W)
		\cup (\{t_3\} \times (V \cup W))
		\cup \left(\bigcup_{1\le i \le q} \{w_i\}\ \times (U_i \cup V_i) \right)$
	(we abuse the notation assuming
	that $\times$ operator
	returns \emph{unordered} pairs minus loops).
	
	If $(G, S)$ is a YES-instance then there exists
	a coloring $c_S$ such that all the constraints in $S$
	are satisfied.
	We extend this coloring to a coloring of $E(G')$ as follows.
	\[
		c(e) = \begin{cases}
			c_S(e) & \mbox{for } e \in E(G)\mbox{, } \\
			1 & \mbox{for } e \in (W \times W) \cup (T \times T) \cup (\{t_3\} \times W)
				\cup \left(\bigcup_{1 \le i \le q} \{w_i\} \times U_i\right)\mbox{, }\\
			0 & \mbox{for } e \in (\{t_2\} \times W) \cup (\{t_3\} \times V)
				\cup \left(\bigcup_{1 \le i \le q} \{w_i\} \times V_i\right).\\
		\end{cases}
	\]
	It suffices to show that all anti-edges of $G'$ are satisfied by the coloring $c$.
	An anti-edge $uv$ inside the set of the vertices $V$ either belongs to $S$
	and it is satisfied by a path in $G$ 
	or it belongs to one of the bicliques $(U_i, V_i; E_i)$
	and then it is satisfied by a path $u w_i v$.
	Inside $W$ and $T$ all the vertices are connected directly.
	An anti-edge $vw$ between $V$ and $W$ is connected by a path $v t_3 w$.
	The vertices of $T$ are connected with $V$ via $\{t_3\}$
	and with $W$ via $\{t_2\}$
	($\{t_3\}$ is also connected directly to $W$).
	So $G'$ is a YES-instance.
	
	If $G'$ is a YES-instance then there exists a coloring $c$ such
	that all the anti-edges in $G'$ are satisfied.
	Note that $S \cap E(G') = \emptyset$.
	An anti-edge belonging to $S$ cannot be satisfied by
	a path using any vertex from $W$ because it is not
	covered by any of the added bicliques.
	It cannot be also satisfied by a path using
	vertex $t_3$ because $t_3$ is the only common neighbor
	of $t_1$ and the vertices of $V$.
	Therefore in $c$ all the
	edges connecting $V$ with $t_3$ have to be in the
	same color, i.e., the color different from $c(t_1 t_3)$.
		Moreover $t_3$ is the only vertex of $T$
		that is adjacent to $V$.
	Hence every anti-edge in $S$ is satisfied using
	only paths inside $G$.
	Therefore $c\mid_V$ is also a coloring satisfying
	an instance $(G, E, S)$.
	
	We added only 	$O(p^2 \Delta_1(G_S) \log |V|)$ new vertices
	and $O((p^2 \Delta_1(G_S)\log |V|)^2  + |V| p^2 \Delta_1(G_S) \log |V|)
	 = O((|V| + p^2 \Delta_1(G_S)\log |V|) \cdot p^2 \Delta_1(G_S) \log |V|)$
	edges.
\end{proof}

\begin{lemma} \label{lem:from-PartRC-to-RC}
	For any fixed $k\ge 3$, given an instance $(G=(V,E), S)$ of \probPartRC
	together with
	a $p$-coloring of the graph $G_S=(V, S)$,
	one can construct an equivalent instance $G'$ of \probRC
	such that
	$|V(G')| = O(
		|V|
		+ k p^2 \Delta_1(G_S) \log |V|
	)$,
	$|E(G')| = O(
		|E(G)|
		+ |V| p^2 \Delta_1(G_S) \log |V|
	)$.
	The construction algorithm can run in $(i)$ expected polynomial time, or $(ii)$ deterministic $2^{|V|}|V|^{O(1)}$ time.
\end{lemma}

\begin{proof}
 In what follows we assume that $G$ contains an isolated vertex $v^*$, for otherwise we can just add it (without changing $S$) and get an equivalent instance of \probPartRC.
 Let us consider a biclique cover of the complement of the graph $G_S$ with $q=O(p^2\Delta_1(G_S)\log |V|)$ bicliques $B_1=(U_1, V_1; E_1), B_2=(U_2, V_2; E_2), \ldots,B_q=(U_q, V_q; E_q)$ as in Lemma~\ref{lem:bliclique_covering}.
 Note that because of the existence of $v^*$, every vertex of $G$ belongs to at least one biclique.
 We construct a graph $G'$ as follows.
 Begin with $G'=G$.
 Next, for every biclique $B_i$, $i=1,\ldots,q$ we add 
 \begin{itemize}
  \item a $2(k-1)$-cycle $C_i=(v_{i,0},v_{i,1},\ldots,v_{i,k-2},w_{i,k-3},\ldots,w_{i,1})$,
  \item an edge $uv_{i,0}$ for every $u\in U_i$,
  \item an edge $vv_{i, k- 2}$ for every $v\in V_i$.
 \end{itemize}
 
 We denote $w_{i, 0}=v_{i, 0}$ and $w_{i,k-2}=v_{i,k-2}$. 
 For every  $i=1,\ldots,q$, the cycle $C_i$ partitions into two paths $P_i=(v_{i,0},v_{i,1},\ldots,v_{i,k-2})$ and $Q_i=(w_{i,0},v_{i,1},\ldots,w_{i,k-2})$.
 Next, we add $2\ceil{\log q}$ vertices $a_1,\ldots,a_{\ceil{\log q}}$ and $b_1,\ldots,b_{\ceil{\log q}}$.
 
 At this point, the construction differs a bit depending on the parity of $k$.
 
 Assume $k$ is odd.
 For every  $i=1,\ldots,q$, the middle vertices of the paths $P_i$ and $Q_i$, i.e., the vertices $v_{i,(k-3)/2},v_{i,(k-1)/2},w_{i,(k-3)/2},w_{i,(k-1)/2}$ are called {\em portals}. For every $t=1,\ldots,\ceil{\log q}$ we put edges between $\{a_t,b_t\}$ and all portals.
 
 Assume $k$ is even.
 Then there are two kinds of portals.
 For every  $i=1,\ldots,q$, the middle vertex of the paths $P_i$ and $Q_i$, i.e., the vertices $v_{i,(k-2)/2},w_{i,(k-2)/2}$ are called {\em $1$-portals}. 
 For every  $i=1,\ldots,q$, the neighbors of the $1$-portals on $P_i$ and $Q_i$, i.e., the vertices $v_{i,(k-4)/2},v_{i,k/2},w_{i,(k-4)/2},w_{i,k/2}$ are called {\em $0$-portals}. 
 For every $t=1,\ldots,\ceil{\log q}$ and for every $i$, we put edges between $\{a_t,b_t\}$ and all $\theta_t(i)$-portals of $C_i$, where $\theta_t(i)$ is the $t$-th bit of $i$.
 
 Finally, form a clique from all vertices $a_{r}$ and $b_r$ for $r=1,\ldots,\ceil{\log q}$.
 This completes the construction. 
 Note that we have added $2(k-1) q + 2\ceil{\log q} = O(kq) = O(k p^2\Delta_1(G_S)\log |V|)$
 vertices and at most
 $q (2(k-1) + |V| + 8 \ceil{\log q} ) + 2 \ceil{\log q}^2
 = O(q|V|) = O(|V|p^2\Delta_1(G_S)\log |V|)$ edges.
 
\begin{figure}[t]

\begin{center}

\begin{tikzpicture}[scale=0.7]

\tikzset{vtx/.style={draw, circle, line width=1pt, inner sep = 1pt}}
\tikzset{edgecolor/.style={draw, rectangle, line width=.5pt, fill=white, pos=0.7, inner sep=1pt}}
\node [vtx] (vi0) at (0,1) {$v_{i,0}$};
\node [vtx] (wi1) at (-1,-2) {$v_{i,1}$};
\node [vtx] (wi2) at (-1,-4) {$v_{i,2}$};
\node [vtx] (vi2) at (1,-4) {$w_{i,2}$};
\node [vtx] (vi1) at (1,-2) {$w_{i,1}$};
\node [vtx] (vi3) at (0,-7) {$v_{i,3}$};

\draw [ thick] (vi0) edge node [edgecolor,pos=0.5] {2} (wi1);
\draw [ thick] (wi1) edge node [edgecolor,pos=0.5] {3} (wi2);
\draw [ thick] (wi2) edge node [edgecolor,pos=0.5] {4} (vi3);
\draw [ thick] (vi3) edge node [edgecolor,pos=0.5] {2} (vi2);
\draw [ thick] (vi2) edge node [edgecolor,pos=0.5] {3} (vi1);
\draw [ thick] (vi1) edge node [edgecolor,pos=0.5] {4} (vi0);

\node [vtx] (vj0) at (6,1) {$v_{j,0}$};
\node [vtx] (wj1) at (5,-2) {$v_{j,1}$};
\node [vtx] (wj2) at (5,-4) {$v_{j,2}$};
\node [vtx] (vj2) at (7,-4) {$w_{j,2}$};
\node [vtx] (vj1) at (7,-2) {$w_{j,1}$};
\node [vtx] (vj3) at (6,-7) {$v_{j,3}$};

\draw [ thick] (vj0) edge node [edgecolor,pos=0.5] {2} (wj1);
\draw [ thick] (wj1) edge node [edgecolor,pos=0.5] {3} (wj2);
\draw [ thick] (wj2) edge node [edgecolor,pos=0.5] {4} (vj3);
\draw [ thick] (vj3) edge node [edgecolor,pos=0.5] {2} (vj2);
\draw [ thick] (vj2) edge node [edgecolor,pos=0.5] {3} (vj1);
\draw [ thick] (vj1) edge node [edgecolor,pos=0.5] {4} (vj0);

\node [vtx, inner sep=2] (a) at (3,1) {$a_t$};
\draw [ thick](wi1) -- node [edgecolor] {3} (a);
\draw [ thick](vi2) -- node [edgecolor] {3} (a);
\draw [ thick](wi2) -- node [edgecolor] {3} (a);
\draw [ thick](vi1) -- node [edgecolor] {3} (a);

\draw [ thick](vj1) -- node [edgecolor] {1} (a);
\draw [ thick](vj2) -- node [edgecolor] {1} (a);
\draw [ thick](wj2) -- node [edgecolor] {1} (a);
\draw [ thick](wj1) -- node [edgecolor] {1} (a);

\node [vtx, inner sep=2] (b) at (3,-7) {$b_t$};
\draw [ thick](a) -- node [edgecolor,pos=0.5] {5} (b);

\draw [ thick](wi1) -- node [edgecolor] {3} (b);
\draw [ thick](vi2) -- node [edgecolor] {3} (b);
\draw [ thick](wi2) -- node [edgecolor] {3} (b);
\draw [ thick](vi1) -- node [edgecolor] {3} (b);

\draw [ thick](vj1) -- node [edgecolor] {5} (b);
\draw [ thick](vj2) -- node [edgecolor] {5} (b);
\draw [ thick](wj2) -- node [edgecolor] {5} (b);
\draw [ thick](wj1) -- node [edgecolor] {5} (b);

\draw [ thick](vi0) -- node [edgecolor,pos=0.5] {1} (-3,2);
\draw [ thick](vi0) -- node [edgecolor,pos=0.5] {1}(-3,1);
\draw [ thick] (vi0) --node [edgecolor,pos=0.5] {1} (-3,0);
\draw [line width=1pt,fill=white] (-3.5,1) ellipse (1 and 2);
\node [scale=2] at (-3.5,1) {$U_i$};

\draw [ thick](vj0) -- node [edgecolor,pos=0.5] {1} (9,2);
\draw [ thick](vj0) -- node [edgecolor,pos=0.5] {1}(9,1);
\draw [ thick] (vj0) --node [edgecolor,pos=0.5] {1} (9,0);
\draw [line width=1pt,fill=white] (9.5,1) ellipse (1 and 2);
\node [scale=2] at (9.5,1) {$U_j$};

\draw [ thick](vi3) -- node [edgecolor,pos=0.5] {5} (-3,-6);
\draw [ thick](vi3) -- node [edgecolor,pos=0.5] {5}(-3,-7);
\draw [ thick] (vi3) --node [edgecolor,pos=0.5] {5} (-3,-8);
\draw [line width=1pt,fill=white] (-3.5,-7) ellipse (1 and 2);
\node [scale=2] at (-3.5,-7) {$V_i$};

\draw [ thick](vj3) -- node [edgecolor,pos=0.5] {5} (9,-6);
\draw [ thick](vj3) -- node [edgecolor,pos=0.5] {5}(9,-7);
\draw [ thick] (vj3) --node [edgecolor,pos=0.5] {5} (9,-8);
\draw [line width=1pt,fill=white] (9.5,-7) ellipse (1 and 2);
\node [scale=2] at (9.5,-7) {$V_j$};

\end{tikzpicture}

\end{center}

\caption{\label{fig:butterfly}Illustration of the $k=5$ case. Rectangular labels denote colors. Here, $t$ is the number of any bit on which $i$ and $j$ differ; $\theta_t(i)=0$ and $\theta_t(j)=1$.}

\end{figure}
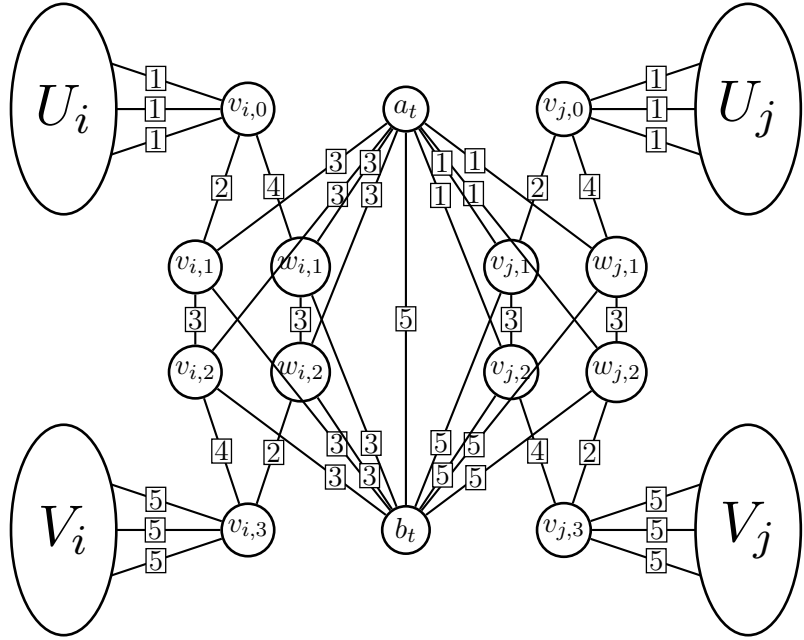

 Assume that $(G,S)$ is a YES-instance of \probPartRC, and let $c$ be the corresponding coloring.
 We will show that there is a rainbow $k$-coloring $c'$ of $E(G')$.
 Define $c'(e)=c(e)$ for $e\in E$.
 Next, for every biclique $B_i$, $i=1,\ldots,q$ we define colors of the corresponding edges as follows.
 \begin{itemize}
  \item The edges of the cycle $(v_{i,0},v_{i,1},\ldots,v_{i,k-2},w_{i,k-3},\ldots,w_{i,1})$ are colored with colors $2,3,\ldots,k-1,2,\ldots,k-1$ respectively.
  \item for every $u\in U_i$, we put $c'(uv_{i,0})=1$,
  \item for every $v\in V_i$, we put $c'(vv_{i, k - 2})=k$.
 \end{itemize}

 Assume $k$ is odd.
 For every $t=1,\ldots,\ceil{\log q}$ and for every  $i=1,\ldots,q$ consider the set $A_{t,i}$ (resp. $B_{t,i}$) of four edges between $a_t$ (resp. $b_t$) and the portals of $C_i$.
 If $\theta_t(i)=0$ the edges of both $A_{t,i}$ and $B_{t,i}$ are colored with $\frac{k+1}2$.
 If $\theta_t(i)=1$ the edges of $A_{t,i}$ are colored with $1$ and the edges of $B_{t,i}$ are colored with $k$.
 
 Assume $k$ is even.
 For every $t=1,\ldots,\ceil{\log q}$ and for every  $i=1,\ldots,q$ consider the set $A_{t,i}$ (resp. $B_{t,i}$) of four or two edges between $a_t$ (resp. $b_t$) and the $\theta_t(i)$-portals of $C_i$.
 If $\theta_t(i)=0$ the edges of both $A_{t,i}$ and $B_{t,i}$ are colored with $\frac{k}2$.
 If $\theta_t(i)=1$ the edges of $A_{t,i}$ are colored with $1$ and the edges of $B_{t,i}$ are colored with $k$.

 Finally, the clique formed of vertices $a_r$ and $b_r$ is colored in color $k$.
 (See Fig~\ref{fig:butterfly} for an illustration of the $k=5$ case.)
 Now we need to verify whether every pair $u,v$ of vertices of $(G',c')$ is connected by a rainbow path. Let us consider cases depending on the types of vertices in the pair.
 \begin{itemize}
  \item If $u,v \in V$ and $\{u,v\} \in S$ then $u$ and $v$ are connected by a rainbow path in $(G,c)$, and we can use the same path.
  \item If $u,v \in V$ and $\{u,v\} \not \in S$ then let $B_i$ be the biclique that contains $uv$ and assume w.l.o.g.\ $u\in U_i$ and $v\in V_i$. 
       Then the path $(u, v_{i,0}, \ldots, v_{i,k-2}, v)$ is colored by $(1,2,\ldots,k-1,k)$ and hence is rainbow.
  \item Assume $u\in V$ and $v\in C_j$ for some $j=1,\ldots,q$.
        Pick any biclique $B_i$ that contains $u$. Assume that $u \in U_i$ (the case $u \in V_i$ is symmetric).
        If $i=j$ we use the path which begins with $uv_{i, 0}$ (colored by $1$) and then continues to $v$ using the shortest path on $C_j$. Since $C_j$ is colored using only colors $2,\ldots,k-1$, this path is rainbow.
        Hence we can focus on the case $i\ne j$.
        Let $t$ be any bit on which $i$ and $j$ differ.
        Assume $k$ is odd. 
        If $v \in \{v_{j,0},\ldots,v_{j,(k-3)/2}\}$, say $v=v_{j,\ell}$, we use the path $(u, w_{i, 0}, \ldots, w_{i,(k-3)/2}, b_{t}, v_{j,(k-3)/2}, \ldots, v_{j, \ell})$.
        Depending on whether $\theta_t(i)$ is $0$ or $1$, the successive edges of this path have colors $(1,k-1,\ldots,\frac{k+3}2,\frac{k+1}2,k,\frac{k-1}2,\ldots,\ell+2)$ or $(1,k-1,\ldots,\frac{k+3}2,k,\frac{k+1}2,\frac{k-1}2,\ldots,\ell+2)$.
        The remaining three cases $v \in \{v_{j,(k-1)/2},\ldots,v_{j,k-2}\}$, $v \in \{w_{j,0},\ldots,w_{j,(k-3)/2}\}$ and $v \in \{w_{j,(k-1)/2},\ldots,w_{j,k-2}\}$ are analogous.
        When $k$ is even, $\theta_t(i)=0$ and $v \in \{v_{j,0},\ldots,v_{j,(k-2)/2}\}$, say $v=v_{j,\ell}$, we use the path $(u, w_{i, 0}, \ldots, w_{i,(k-4)/2}, b_{t}, v_{j,(k-2)/2}, \ldots, v_{j, \ell})$.
        Again, there are three more analogous cases when $\theta_t(i)=0$ and four ones when $\theta_t(i)=1$.
        When $u \in V_i$ the paths are analogous, but we use $a_{t}$ instead of $b_{t}$, to avoid repeating the color $k$.
  \item If $u\in V$ and $v=a_r$ or $v=b_r$ for some $r=1,\ldots,\ceil{\log q}$, the requested path is a subpath of one of the rainbow paths described in the previous case.
  \item Assume $u\in C_i$ and $v\in C_j$ for some $i,j=1,\ldots,q$.      
        When $i=j$ we can reach $v$ from $u$ by a rainbow path going through the shortest path in $C_i$. 
        Assume $i\ne j$.
        Let $t$ be the first bit on which $i$ and $j$ differ.
        Let us describe the case of odd $k$ only (the case of even $k$ is similar, but slightly different, because of non-symmetric neighborhoods of vertices $a_r$ and $b_r$).
        Denote $\clone_j(v_{i,r}) = v_{j,r}$ and $\clone_j(w_{i,r}) = w_{j,r}$.
        Recall that for every path $v_{i,0},\ldots,v_{i,k-2}$ there are two portals, similarly for every path $w_{i,0},\ldots,w_{i,k-2}$.
        Two portals $x_1$ and $x_2$ on the same path are called {\em twins} and we denote $\twin(x_1)=x_2$ and $\twin(x_2)=x_1$.
        If $v_{i,r}$ is a portal then $w_{i,r}$ is also a portal and we denote $\opp(v_{i,r}) = w_{i,r}$ and $\opp(w_{i,r}) = v_{i,r}$.
        First assume the shortest path $P$ from $u$ to $\clone_i(v)$ in $C_i$ goes through a portal. Let $x$ be the first portal visited by $P$ from $u$.
        Then the rainbow path from $u$ to $v$ is formed by going in $P$ from $u$ to $x$, then going to $a_{t}$, and then either through $\twin(\clone_j(x))$ to $v$ using the shortest path on $C_j$, when $v\ne\clone_j(x)$, or directly to $\clone_j(x)$, when $v=\clone_j(x)$.
        The colors are the same as on path $P$, plus color $1$ (and plus color $\frac{k+1}2$ when $v=\clone_j(x)$), so the path is rainbow.
        Now assume that $P$ does not go through a portal. 
        Then the rainbow path from $u$ to $v$ is formed by going in $P$ from $u$ to the nearest portal $x$, then going through $a_{t}$ to $\clone_j(\opp(x))$, and then to $v$ using the shortest path on $C_j$. This path uses colors of $[k]\setminus c(E(P)) \cup \{1\}$, each exactly once.
  \item If $u\in C_i$ for some $i$ and $v=a_r$ for some $r=1,\ldots,\ceil{\log q}$, the requested path is a subpath of one of the rainbow paths described in the previous case; when $v=b_r$ we use the path to $a_r$ and extend it by the edge $a_rb_r$.
  \item If $\{u,v\} \subseteq \{a_r\ :\ r=1,\ldots,\ceil{\log q}\} \cup \{b_r\ :\ r=1,\ldots,\ceil{\log q}\}$, then $u$ and $v$ are adjacent, hence connected by a rainbow path of length 1.
 \end{itemize}

 Now assume that $G'$ is a YES-instance of \probRC, and let $c'$ be the corresponding coloring.
 We claim that the coloring $c'|_{E(G)}$ of $E(G)$ satisfies all the pairs in $S$.
 It follows from the observation that for every $\{u,v\}\in S$ every path between $u$ and $v$ that leaves $E(G)$ either has length at least $k+1$, or contains two edges $av_{i,0}$ and $v_{i,0}b$ for some $i=1,\ldots,q$ and $a,b \in U_i$, or contains two edges $av_{i,3}$ and $v_{i,3}b$ for some $i=1,\ldots,q$ and $a,b \in V_i$.
\end{proof}

\subsection{Putting everything together}

By pipelining lemmas~\ref{lem_transformation},~\ref{lem:from-TSAT-to-PartRTwoCExt},~\ref{lem:from-PartRTwoCExt-to-PartRCExt}, and~\ref{lem:from-PartRCExt-to-PartRC} we get the following corollary.

\begin{corollary} \label{lem:from-TSAT-to-PartRTwoC}
Fix $k\ge 2$.
	Given a \probTSAT formula $\varphi$ with $m$ clauses
	one can construct in polynomial time an equivalent instance $(G=(V,E), S)$ of \probPartRC
	such that
	$|V| = O(m^{2/3})$,
	$|E| = O(m)$,
	$\Delta((V, S)) = O(m^{1/3})$, and
	the graph $G_S=(V,S)$ is $O(1)$-colorable.
\end{corollary}

Note that in Corollary~\ref{lem:from-TSAT-to-PartRTwoC} we have $|S|=|V|\Delta((V, S)) = O(m)$.
It follows that the Sparsification Lemma (Lemma~\ref{lem_sparsification}) and Corollary~\ref{lem:from-TSAT-to-PartRTwoC} imply the following result.

\begin{theorem}
	For any $k\ge 2$, \probPartRC cannot be solved in time $2^{o(n^{3/2}) + o(m) + o(s)}$ where $n$ is the number
	of vertices, $m$ is the number of edges, and $s$ is the number of requests, unless ETH fails.
\end{theorem}

Pipelining Corollary~\ref{lem:from-TSAT-to-PartRTwoC} and Lemma~\ref{lem:from-PartRTwoC-to-RTwoC} (for $k=2$) or Lemma~\ref{lem:from-PartRC-to-RC} (for $k\ge 3$) gives the following corollary.

\begin{corollary} \label{lem:from-TSAT-to-RTwoC}
Fix $k\ge 2$.
	Given a \probTSAT formula $\varphi$ with $O(m)$ clauses
	one can construct an equivalent instance $G$ of \probRC
	with $O(m^{2/3})$ vertices and $O(m \log m)$ edges.
	The construction algorithm can run in $(i)$ expected polynomial time, or $(ii)$ deterministic $2^{O(m^{2/3})}$ time.
\end{corollary}

Again, the above and the Sparsification Lemma immediately imply the following hardness result.
(Note that we do not use the randomized reduction algorithm --- we state it just in case it is useful in some other applications.)

\begin{theorem}
\label{thm:RC-3/2}
	For any $k\ge 2$, \probRC cannot be solved in time $2^{o(n^{3/2}) + o(m/\log m)}$ where $n$ is the number
	of vertices and $m$ is the number of edges, unless ETH fails.
\end{theorem}

%

\section{Algorithms for \probPartRC}
\label{sec:alg-subset}

In this section we study FPT algorithms for \probPartRC parameterized by $|S|$. 
We provide two such algorithms, based on different approaches: one for $k=2$ case, and one (slightly slower) for the general case.
Consider an instance $(G,S)$ of the \probPartRC problem.
Note that we can assume that $S\subseteq \barE$, since any constraint $\{u,v\}\in E$ is satisfied in every edge coloring. 
Moreover, we say that a pair $\{u,v\}$ is {\em feasible} when the distance between $u$ and $v$ is at most $k$.
The set of all feasible pairs is denoted by $F(G)$.
Clearly, when $S$ contains a request which is not feasible, then $(G,S)$ is a trivial NO-instance.
Hence, throughout this section we assume $S\subseteq \barE \cap F(G)$.

\subsection{The $k=2$ case}
\label{sec:in-ex}

For any $X\subseteq S$ let $\Pp_X$ be the set of all 2-edge paths between the pairs of vertices in $X$.
Denote $E(\Pp_X) = \bigcup_{P\in \Pp_X}E(P)$.
For two edges $e_1,e_2\in E(G)$ we say that $e_1$ and $e_2$ are {\em linked by $X$}, denoted as $e_1 \sim_X e_2$ when there are two paths $P_1,P_2\in \Pp_X$ (possibly $P_1=P_2$) such that $e_1\in E(P_1)$, $e_2\in E(P_2)$ and $E(P_1) \cap E(P_2) \ne \emptyset$. 
Let $\approx_X$ be the transitive closure of $\sim_X$.
Then $\approx_X$ is an equivalence relation.
Recall that $E(G) / \approx_X$ denotes the quotient set of the relation $\approx_X$.

The main observation of this section is the following theorem.

\begin{theorem}
\label{thm:num-2-cols}
The number of 2-colorings of $E(G)$ that satisfy all the pairs in $S$ is equal to
\begin{equation}
\label{eq:num-2-cols}
\sum_{X\subseteq S}(-1)^{|X|}2^{|E(G) / \approx_X|}.
\end{equation}
\end{theorem}

In the proof we make use of the well-known inclusion-exclusion principle. 
Below we state it in the intersection version (see, e.g.,~\cite{fptbook})

\begin{theorem}[Inclusion--exclusion principle, intersection version]
\label{thm:ie}
Let $A_1,\ldots,A_n\subseteq U$, where $U$ is a finite set. 
Denote $\bigcap_{i\in\emptyset}(U\setminus A_i)=U$. 
  Then
   \[\big|\bigcap_{i\in[n]}A_i\big|=\sum_{X\subseteq[n]}(-1)^{|X|}{\big|\bigcap_{i\in X}(U\setminus{A_i})\big|}.\]
\end{theorem}

\begin{proof}[Proof of Theorem~\ref{thm:num-2-cols}]
Let us define, for every pair $\{u,v\}\in S$ (say, $u<v$), the set $A_{u,v}$ of 2-edge colorings of $G$ that satisfy $\{u,v\}$.
Note that the number of rainbow 2-colorings of $G$ that satisfy all the pairs in $S$ is equal to $|\bigcap_{\{u,v\}\in S}A_{u,v}|$. 
By Theorem~\ref{thm:ie} it suffices to show that, for any subset $X\subseteq S$, the number $\#_X$ of 2-colorings such that {\em none} of the pairs in $X$ is satisfied, equals $2^{|E(G) / \approx_X|}$.

 Fix any coloring $c$ that does not satisfy any pair from $X$.
 Then every path from $\Pp_X$ has both edges of the same color.
 Hence, for two edges $e_1,e_2\in E(G)$, if $e_1 \sim_X e_2$ then $e_1$ and $e_2$ are colored by $c$ with the same color.
 It follows that for any equivalence class $A$ of $\approx_X$, all edges of $A$ are have the same color in $c$.
 This proves that $\#_X \le 2^{|E(G) / \approx_X|}$.
 
 For every function $c_0 : (E(G) / \approx_X) \rightarrow \{1,2\}$ we can define the coloring $c : E(G) \rightarrow \{1,2\}$ by putting $c(e) = c_0([e]_{\approx_X})$ for every edge $e\in E(G)$. (Note that the edges that do not belong to any path in $\Pp_X$ form singleton equivalence classes.)
 Then, $c$ does not satisfy any pair from $X$, because if some pair $\{u,v\}$ is satisfied then there is a 2-color path $uxv$; but $ux \sim_X xv$, so $[ux]_{\approx_X} = [xv]_{\approx_X}$ and $c(ux)=c(xv)$, a contradiction.
 It follows that $\#_X \ge 2^{|E(G) / \approx_X|}$. This ends the proof.
\end{proof}

\begin{corollary}
\label{cor:alg-num-2-cols}
For any graph $G=(V,E)$ and a set of requests $S$ the number of 2-colorings of $E$ that satisfy all the pairs in $S$ can be computed in $O(2^{|S|}(|E|+|S|\cdot|V|))$ time and polynomial space.
In particular, \probPartRTwoC can be decided within the same time.
\end{corollary}

\begin{proof}
 Fix any set $X\subseteq S$.
 The relation $\approx_X$ can be computed in $O(|E|+|S|\cdot|V|)$ time as follows. 
 We use the standard union-find data structure for maintaining disjoint sets under union operations.
 We begin with $|E|$ singleton sets, one per edge.
 Next, we consider paths in $\Pp_X$, one by one. 
 For every such path $uxv\in \Pp_X$ we perform Union operation on the sets containing $ux$ and $xv$.
 The final sets correspond to the equivalence classes of $\approx_X$.
 Note that $|\Pp_X| \le |S|\cdot |V|$ and we can enumerate elements of $|\Pp_X|$ in $O(|S|\cdot |V|)$ time: for every pair $\{u,v\}\in S \cap X$ we mark the neighbors of $u$; then every marked neighbor of $v$ corresponds to a path in $\Pp_X$.
 The $O(|E|+|S|\cdot|V|)$ time bound follows from the well-known fact that $m$ union operations (without any find operations) implemented using both path compression and union-by-rank take only $O(m)$ time.
 
 Now, by Theorem~\ref{thm:num-2-cols} we can compute the number of 2-colorings of $E$ that satisfy all the pairs in $S$ using only $2^{|S|}$ additions and subtractions of powers of two.
 The exponent of each such power is found in time $O(|E|+|S|\cdot|V|)$, as explained above. 
 Each such addition/subtraction can be performed in $O(1)$ amortized time, we skip the easy details.
\end{proof}

Let us remark here that the algorithm from Corollary~\ref{cor:alg-num-2-cols} only decides whether the coloring exists, without finding it.
However, one can find the coloring at the cost of $O(|E|)$ multiplicative overhead in the running time.
Indeed, the algorithm from Corollary~\ref{cor:alg-num-2-cols} can be easily modified (without affecting the asymptotic running time bound) so that it finds the number of 2-colorings of $E$ that satisfy all the pairs in $S$, and that {\em extend} a given partial coloring $c_0 : E \rightarrow \{1,2\}$. 
Then we can find the requested coloring by staring from empty partial coloring and next extending it edge by edge, always choosing a color that gives a positive number of extensions.

\subsection{The general case}

In this section we use partial colorings. 
For convenience, a partial coloring is represented as a function $c:E\rightarrow [k] \cup \{\bot\}$, where the value $\bot$ corresponds to an uncolored edge.
By $\Dom(c)$ we denote the domain of the corresponding partial function, i.e., $\Dom(c) = c^{-1}([k])$.
The partial coloring which does not color anything, i.e., is constantly equal to $\bot$ is denoted by $c_\bot$.

For a graph $G=(V,E)$ consider a partial edge coloring $c:E\rightarrow [k]\cup\{\bot\}$.
A {\em guide function} is any function of the form $f : S \rightarrow {\Dom(c) \choose \le k}$, i.e., any function that assigns sets of at most $k$ colored edges to all requests in $S$.
A constant guide function equal to $\emptyset$ for every request in $S$ is denoted by $g_{S,\emptyset}$.
Pick any pair $\{u,v\} \in S$.
We say that a walk $W$ connecting $u$ and $v$ is $f$-guided if every color appears at most once on $W$, and $f(\{u,v\}) \subseteq E(W)$.
We say that a coloring $c$ is $(f,S)$-rainbow when for every pair $\{u,v\}\in S$ there is an $f$-guided walk between $u$ and $v$.
Note that $(G,S)$ is a YES-instance of \probPartRC iff there is an $(g_{S,\emptyset},S)$-rainbow coloring. Indeed, every rainbow walk contains a rainbow path.

The following lemma is going to be useful in our branching algorithm.

\begin{lemma}
\label{lem:alg-f-guided}
 Let $G=(V,E)$ be a graph, and let $S$ be a set of requests. 
 Let $c_0:E\rightarrow [k]$ be a partial edge coloring and let $f : S \rightarrow {\Dom(c_0) \choose \le k}$ be a guide function.
 Then, given a pair $\{u,v\}\in S$ in time $2^kn^{O(1)}$ one can find an $f$-guided $u$-$v$ walk of length at most $k$, if it exists.
\end{lemma}

\begin{proof}
 The algorithm is as follows.
 We can assume that $f(\{u,v\})$ does not contain two edges of the same color, for otherwise the requested walk does not exist.
 For every $e\in f(\{u,v\})$ we remove all the edges of color $c_0(e)$.
 Next, we put back edges of $f(\{u,v\})$.
 Then it suffices to find in the resulting graph $G'$ any $u$-$v$ path of length at most $k$ and with no repeated colors that visits all the colors of the edges in $f(\{u,v\})$.
 This is done using dynamic programming.
 For every vertex $x\in V$, subset $X\subseteq [k]$ and integer $\ell = 0,\ldots,k$ we find the boolean value $T[x,X,\ell]$ which is true iff there is a $u$-$x$ walk of length $\ell$ which does not repeat colors and visits all the colors from $X$, but not more. 
 We initialize $T[u,\emptyset,0]=\mathtt{true}$ and $T[x,\emptyset,0]=\mathtt{false}$ for every $x\ne u$.
 Next we iterate through the remaining triples $(x,X,\ell)$, in the nondecreasing order of $\ell$ and $X$'s cardinalities.
 The value of $T[x,X,\ell]$ is then computed using the formula 
 \[T[x,X,\ell] = \bigvee_{yx\in c_0^{-1}(X\cup\{\bot\}) \cap E(G')}T[y,X\setminus\{c_0(yx)\},\ell-1].\]
 The requested walk exists iff $T[v,X,\ell]=\mathtt{true}$ for any $\ell=0,\ldots,k$ and $X$ such that $c_0(f(\{u,v\}))\subseteq X$.
 The walk is retrieved using standard DP methods.	
\end{proof}

Now we are ready to describe our branching algorithm. 
Let $(G=(V,E),S)$ be the input instance. Our algorithm consists of a recursive procedure {\sc FindColoring} which gets three parameters: $S_0$ (a set of requests), $c_0:E\rightarrow [k]\cup\{\bot\}$ (a partial coloring), and a guide function $f : S \rightarrow {\Dom(c_0) \choose \le k}$. It is assumed that for every request $r\in S$, every pair of different edges $e_1,e_2\in f(r)$ is colored differently by $c_0$. The goal of the procedure {\sc FindColoring} is to find an $(f,S_0)$-rainbow coloring $c:E\rightarrow [k]$ which extends $c_0$. 
Thus the whole problem is solved by invoking $\text{\sc FindColoring}(S,c_\bot,g_{S,\emptyset})$.
A rough description of {\sc FindColoring} is as follows.
We pick any pair $\{u,v\}\in S_0$ and we find any $f$-guided $u$-$v$ walk $W$ of length at most $k$ using Lemma~\ref{lem:alg-f-guided}.
Let $c_1$ be obtained from $c_0$ by coloring the uncolored edges of $W$ to get a rainbow walk.
If $\text{\sc FindColoring}(S_0\setminus\{u,v\},c_1,f|_{S_0\setminus\{u,v\}})$ returns a coloring, we are done.
But if no such coloring exists then we know that we made a wrong decision: coloring some of the uncolored edges $e$ of $W$ into $c_1(e)$ (instead of some color $\alpha$) makes some other request $r\in S_0 \setminus \{\{u,v\}\}$ impossible to satisfy.
For every possible triple $(e,\alpha,r)$ we invoke {\sc FindColoring} with the same set of requests $S_0$, partial coloring $c_0$ extended by coloring $e$ with $\alpha$, and the guide function $f$ extended by putting $f(r) := f(r) \cup \{e\}$.

{
\begin{algorithm}[t]
\caption{\textsc{FindColoring}$(S_0,c_0,f)$}
\label{alg:findcoloring}
\footnotesize
  \If{$S_0 = \emptyset$}{
    \Return{$c_0$}
  }  
  \If{for some $r\in S_0$ there are edges $e_1,e_2\in f(r)$ with $c_0(e_1)=c_0(e_2)$}{
    \Return{$\mathtt{null}$}
  }  
  Pick any $\{u,v\}\in S_0$\;
  Find any $f$-guided $u$-$v$ walk $W$ of length at most $k$ using Lemma~\ref{lem:alg-f-guided}\;
  \If{$W$ does not exist}{
    \Return{$\mathtt{null}$}
  }  
  Let $c_1$ be obtained from $c_0$ by coloring the uncolored edges of $W$ to get a rainbow walk\;
  {\bf if} $\text{\sc FindColoring}(S_0\setminus\{u,v\},c_1,f|_{S_0\setminus\{u,v\}})\ne \mathtt{null}$ {\bf then return} the coloring found\;
  \For{$e\in E(W) \setminus \Dom(c_0)$}{
    \For{$\alpha\in [k]$}{
      \For{$r\in S_0 \setminus \{\{u,v\}\}$}{
        Let $c_{e,\alpha}$ be obtained from $c_0$ by coloring $e$ with $\alpha$\;
        Let $f_{e,r}$ be obtained from $f$ by putting $f(r) := f(r) \cup \{e\}$\;
        {\bf if} $\text{\sc FindColoring}(S_0,c_{e,\alpha},f_{e,r})\ne \mathtt{null}$ {\bf then return} the coloring found\;\label{line:2nd-recurse}
      }
    }
  }  
  \Return{$\mathtt{null}$}
\end{algorithm}
}

A precise description of procedure {\sc FindColoring} can be found in Pseudocode~\ref{alg:findcoloring}. 
The following lemma proves its correctness.

\begin{lemma}
\label{lem:FindColoring}
 Procedure {\sc FindColoring} invoked with parameters $(S_0,c_0,f)$ finds an $(f,S_0)$-rainbow coloring $c:E\rightarrow [k]$ which extends $c_0$, whenever it exists. 
\end{lemma}

\begin{proof}
The proof is by induction on the sum of $|S_0|$ and the number of uncolored edges.
It is clear that if $|S_0|=0$ or all the edges are colored then the algorithm behaves correctly.
In the induction step, the only non-trivial thing to check is whether any of the calls in lines 10 or 16 returns a coloring, provided that there is a solution, i.e., an $(f,S_0)$-rainbow coloring $c:E\rightarrow [k]$ which extends $c_0$.
Assume that no coloring is returned in Line 16.
Then for every edge $e\in E(W) \setminus \Dom(c_0)$, and request $r\in S_0 \setminus \{\{u,v\}\}$ coloring $c$ is not a $(f_{e,r},S_0)$-rainbow coloring, for otherwise the call $\text{\sc FindColoring}(S_0,c_{e,c(e)},f_{e,r})$ returns a coloring.
If follows that for every edge $e\in E(W) \setminus \Dom(c_0)$ and request $r\in S_0 \setminus \{\{u,v\}\}$ the walk that realizes the request $r$ in the coloring $c$ does not contain $e$.
Hence, the following coloring
\[
c'(e)=\begin{cases}
    c(e) & \text{if $e\not\in E(W)$,}\\
    c_1(e) & \text{if $e\in E(W)$.}\\
   \end{cases}
\]
is another $(f,S_0)$-rainbow coloring, and it extends $c_1$. 
It follows that the call in Line 10 returns a coloring, as required.
\end{proof}

\begin{theorem}
\label{thm:partRC-is-FPT}
 For every integer $k$, there is an FPT algorithm for \probPartRC parameterized by $|S|$.
 The algorithm runs in time $(k^2|S|)^{k|S|}2^kn^O(1)$, in particular in $|S|^{O(|S|)}n^{O(1)}$ time for every fixed $k$.
\end{theorem}

\begin{proof}
By Lemma~\ref{lem:FindColoring} \probPartRC is solved by invoking $\text{\sc FindColoring}(S,c_\bot,g_{S,\emptyset})$.
Note that whenever we go deeper in the recursion either some request of $S_0$ gets satisfied, or $|f(r)|$ increases for some $r\in S_0$.
When $|f(r)|$ increases to $k+1$, the corresponding recursive call returns $\mathtt{null}$ immediately (because the condition in Line 3 holds).
It follows that the depth of the recursion is at most $|S|k$. 
Since in every call of \probPartRC the algorithm uses time $2^kn^{O(1)}$ (by Lemma~\ref{lem:FindColoring}) and branches into at most $1+k^2(|S|-1)\le k^2|S|$ recursive calls, the total time is $(k^2|S|)^{k|S|}2^kn^O(1)$, as required.
\end{proof}

\section{Algorithms for \probMaxRC}
\label{sec:alg:max}

We start from observing a simple approximation algorithm for \probMaxRC. 
Actually, it works also for a more general problem, which we call \probMaxPartRC.
In this problem we are given an input as in \probPartRC, i.e., graph $G$ and a set of anti-edges $S$, and an additional parameter $q\in\mathbb{N}$.
The goal is to find a $k$-coloring of $E(G)$ that satisfies at least $q$ pairs of $S$. 
Thus, \probMaxPartRC is a generalization of \probMaxRC studied by Ananth \emph{et al.}~\cite{Ananth:fsttcs11}, where $S = \barE$.
Recall from the previous section that $F(G)$ denotes the set of feasible pairs of vertices in $G$, i.e., pairs at distance at most $k$.

\begin{proposition}
\label{prop:approx}
 If all the pairs in $S$ are feasible, then for every $k \ge 2$ \probMaxPartRC admits a deterministic polynomial time algorithm which finds a coloring that satisfies at least $\frac{k!}{k^k} |S|$ pairs from $S$.
 In particular, this is a $k!/k^k$-approximation algorithm.
\end{proposition}

\begin{proof}
 Consider an algorithm which returns a random $k$-edge-coloring $c$. 
 Fix any anti-edge $\{u,v\}\in S$. 
 Since $\{u,v\}$ is feasible we can pick a path $P$ of length at most $k$ between $u$ and $v$.
 Let $\ell = |E(P)|$.
 The probability that $P$ is rainbow equals $\frac{k^{\underline{\ell}}}{k^\ell} \ge \frac{k!}{k^k}$.
 Hence, the probability that $uv$ is satisfied is at least $\frac{k!}{k^k}$.
 By the linearity of expectation, it follows that the expected number of satisfied anti-edges from $S$ is at least $\frac{k!}{k^k} |S|$.
 
 Using the standard method of conditional expectation (see e.g.,~\cite{Vazirani-book}) we can derandomize the algorithm above so that it also runs in polynomial time and always returns a coloring which satisfies at least $\frac{k!}{k^k} |S|$ anti-edges from $S$. (Note: it might be hard to compute exact conditional expected value of the number of satisfied anti-edges. Instead, it is sufficient to choose one path of length at most $k$ between every pair in $S$ and compute exact expected value of the chosen paths which became rainbow.)
 Since in the optimum solution at most $|S|$ anti-edges from $S$ are satisfied, the claim follows.
\end{proof}

Ananth \emph{et al.}~\cite{Ananth:fsttcs11} showed that \probMaxRTwoC has a kernel of $4q$ vertices, thus proving that the problem is FPT.
They do not state the running time of the resulting FPT algorithm, but since the best known rainbow 2-coloring algorithm runs in $2^{|E(G)|}n^{O(1)}$ time, we can conclude that it implies an FPT algorithm running in time $2^{8q^2}n^{O(1)}$.
We can get a different, faster, and more general (i.e., solving \probMaxRC for any $k\ge 3$) FPT algorithm by combining Theorem~\ref{thm:partRC-is-FPT} and Corollary~\ref{cor:alg-num-2-cols}.
For $k=2$ the algorithm gets even faster if we replace Theorem~\ref{thm:partRC-is-FPT} by Proposition~\ref{prop:approx}. The claim follows.

\begin{theorem}
 For any fixed $k\ge 2$, \probMaxRC parameterized by $q$ is in FPT and can be solved in $2^{O(q\log q)}n^{O(1)}$ deterministic time and polynomial space.
 For the special case of $k=2$ it can be done in time $8^q n^{O(1)}$.
\end{theorem}

\begin{proof}
 Our algorithm first verifies whether $q\le \frac{k^k}{k!}|\barE \cap F(G)|$ and if that is the case it answers YES, which is justified by Proposition~\ref{prop:approx} (with $S=\barE \cap F(G)$).
 Otherwise $|\barE \cap F(G)| < \frac{k!}{k^k}q = O(q)$.
 Then the algorithm just applies Theorem~\ref{thm:partRC-is-FPT} for every subset $S \in {\barE \cap F(G) \choose q}$.
 Since there are $O(2^{|\barE \cap F(G)|})=2^{O(q)}$ such subsets, and each of them is processed in $2^{O(q\log q)} n^{O(1)}$ time by Theorem~\ref{thm:partRC-is-FPT}, the claim follows.
 For the special case of $k=2$, there are $O(2^{|\barE \cap F(G)|})=O(2^{2q})$ such subsets, and each of them is processed in $2^q n^{O(1)}$ time by Corollary~\ref{cor:alg-num-2-cols}, which gives the total time of $8^q n^{O(1)}$.
\end{proof}

%

The linear kernel of Ananth \emph{et al.}~\cite{Ananth:fsttcs11} for $k=2$ raises a natural question whether there is a polynomial kernel for larger $k$.
We answer this question in the affirmative, by using a different approach (namely, Proposition~\ref{prop:approx} again).

%
\begin{theorem}
For any $k \ge 2$, \probMaxRC has a kernel of size $O(q)$ when parameterized by $q$.
\end{theorem}
\begin{proof}
If $q\le k!/k^k |\barE \cap F(G)|$, the algorithm can answer YES by Proposition~\ref{prop:approx}. 
Hence, in what follows we assume $|\barE \cap F(G)| < qk^k / k!$. Define

\[V_1 = \{u\ :\ \text{there is a vertex $v$ such that $\{u,v\} \in \barE \cap F(G)$}\}.\]

Let $V_2 = V \setminus V_1$.
Now consider any connected component $H$ of $G$.
Denote $H_1 = V_1 \cap V(H)$ and $H_2 = V_2 \cap V(H)$.
We claim that every pair of different vertices $x \in H_2$ and $y \in V(H)$ is adjacent.
Assume for the contradiction, that there is a pair $x \in H_2$ and $y \in V(H)$ such that $xy\not\in E(G)$ and among such pairs $\dist_G(x,y)$ is minimal.
Denote $d=\dist_G(x,y)$.
Let $x=x_0, x_1, \ldots, x_d=y$ be a path between $x$ and $y$ of length $d$. 
By minimality of $\dist_G(x,y)$, $xx_{d-1}\in E(G)$.
It follows that $d=2$.
Hence, $\{x,y\}\in \barE \cap F(G)$, which implies $x \in H_1$, a contradiction.

Our kernelization algorithm verifies whether there is a connected component $H$ such that $|H_2| \ge |\barE \cap F(G) \cap {V(H)\choose 2}|$.
If this is the case, we can satisfy every feasible anti-edge in $H$. 
Indeed, for every such anti-edge $\{u,v\}$ we choose a different vertex $x_{uv}$ in $H_2$.
Recall that by the claim proved above $x_{uv}$ is adjacent to both $u$ and $v$.
Thus every feasible anti-edge in $H$ gets its $2$-vertex path, and the paths are edge disjoint (even vertex disjoint), so we can color each of them independently into two colors. Hence the algorithm can remove $H$ from the graph and decrease $q$ by $\min\{q,|\barE \cap F(G) \cap {V(H)\choose 2}|\}$, obtaining an equivalent instance. This process is continued until no such complement exists. Let $C(G)$ be the set of connected components of $G$. 
Then \[|V_2|=\sum_{H\in C(G)}|H_2| < \sum_{H\in C(G)}\left|\barE \cap F(G) \cap {V(H)\choose 2}\right| = |\barE \cap F(G)|.\]
By the definition of $V_1$, we get also $|V_1| \le 2|\barE \cap F(G)|$. It follows that $|V| \le 3|\barE \cap F(G)| \le 3qk^k / k! = O(q)$, as required.
\end{proof}



\section{Further Work}
\label{sec:further}

We believe that this work only initiates the study of fine-grained complexity of variants of \probRC.
In particular, many open questions are still unanswered.
The ultimate goal is certainly to get tight bounds. 
We pose the following two conjectures.

\begin{conjecture}
\label{con:2^E}
For any integer $k\ge 2$, there is no $2^{o(|E|)}n^{O(1)}$-time algorithm for \probRC, unless ETH fails.
\end{conjecture}

\begin{conjecture}
\label{con:2^V2}
For any integer $k\ge 2$, there is no $2^{o(n^2)}n^{O(1)}$-time algorithm for \probRC, unless ETH fails.
\end{conjecture}

Note that in this work we have settled Conjecture~\ref{con:2^E} for \probPartRC, and for \probRC we showed a slightly weaker, $2^{o(|E|/\log|E|)}n^{O(1)}$ bound.
However, avoiding this $\log|E|$ factor seems to constitute a considerable technical challenge.

In this paper we gave two algorithms for \probPartRC parameterized by $|S|$, one working in $2^{|S|}n^{O(1)}$ time for $k=2$ and another, working in time $|S|^{O(|S|)}n^{O(1)}$ for every fixed $k$.
We conjecture that there exists an algorithm running in time $2^{O(|S|)}n^{O(1)}$ for every fixed $k$.

Finally, we would like to propose yet another parameterization of \probRC.
Assume we are given a graph $G=(V,E)$ and a subset of vertices $S\subseteq V$.
In the \probSteinerRC problem the goal is to determine whether there is a rainbow $k$-coloring such that every pair of vertices in $S$ is connected by a rainbow path.
By our Theorem~\ref{thm:RC-3/2}, \probSteinerRC has no algorithm running in time $2^{o(|S|^{3/2})}$, under ETH.
On the other hand, our algorithm for \probPartRC implies that \probSteinerRC parameterized by $|S|$ admits an FPT algorithm with running time of $2^{O(|S|^2\log |S|)}n^{O(1)}$.
It would be interesting make the gap between these bounds smaller.

\bibliographystyle{abbrv}

\end{document}